\PassOptionsToPackage{dvipsnames}{xcolor}
\documentclass[11pt]{article}
\usepackage{amssymb,amsmath,amsthm,multirow,color,xcolor,cancel,bbm,tikz,mathrsfs,hyperref}
\usepackage[capitalise]{cleveref}
\usepackage[textheight=24cm,textwidth=16cm]{geometry}
\usetikzlibrary{shapes.arrows,backgrounds}

\hypersetup{
    colorlinks=true,
    linkcolor=blue,
    citecolor=cyan,
    filecolor=magenta,
    urlcolor=green,
    pdftitle={Attractor separation and signed cycles in asynchronous Boolean networks},
}

\newtheorem{proposition}{Proposition}
\newtheorem{lemma}{Lemma}
\newtheorem{theorem}{Theorem}

\newtheorem{remark}{Remark}
\newtheorem{example}{Example}

\newtheorem{conjecture}{Conjecture}

\newcommand{\BF}[1]{{{\boldmath{\bf #1}\unboldmath}}}
\newcommand{\EM}[1]{{\em\textcolor{Maroon}{#1}}}
\newcommand{\EMM}[1]{\textcolor{Maroon}{#1}}
\newcommand{\B}{\{0,1\}}
\newcommand{\ONE}{\mathbf{1}}
\newcommand{\ZERO}{\mathbf{0}}

\title{Attractor separation and signed cycles\\ in asynchronous Boolean networks}

\date{\today}

\author{
Adrien Richard\footnote{Universit\'e Côte d’Azur, CNRS, I3S, Sophia Antipolis, France. \newline \indent\indent{\tt  adrien.richard@cnrs.fr}}~~and 
Elisa Tonello\footnote{Department of Mathematics and Computer Science, Freie Universit\"at Berlin, Germany.\newline \indent\indent{\tt elisa.tonello@fu-berlin.de}}
}

\begin{document}

\maketitle


\begin{abstract}
The structure of the graph defined by the interactions in a Boolean network can determine properties of the asymptotic dynamics. 
For instance, considering the asynchronous dynamics, the absence of positive cycles guarantees the existence of a unique attractor, and the absence of negative cycles ensures that all attractors are fixed points.
In presence of multiple attractors, one might be interested in properties that ensure that attractors are sufficiently ``isolated'', that is, they can be found in separate subspaces or even trap spaces, subspaces that are closed with respect to the dynamics.
Here we introduce notions of separability for attractors and identify corresponding necessary conditions on the interaction graph. In particular, we show that if the interaction graph has at most one positive cycle, or at most one negative cycle, or if no positive cycle intersects a negative cycle, then the attractors can be separated by subspaces. If the interaction graph has no path from a negative to a positive cycle, then the attractors can be separated by trap spaces.
Furthermore, we study networks with interaction graphs admitting two vertices that intersect all cycles, and show that if their attractors cannot be separated by subspaces, then their interaction graph must contain a copy of the complete signed digraph on two vertices, deprived of a negative loop. We thus establish a connection between a dynamical property and a complex network motif.
The topic is far from exhausted and we conclude by stating some open questions.

\end{abstract}

\section{Introduction}

A \EM{Boolean network} (BN) is a finite dynamical system usually defined by a function  
\[
f:\B^n\to\B^n,\qquad x=(x_1,\dots,x_n)\mapsto f(x)=(f_1(x),\dots,f_n(x)).
\]

BNs have many applications. In particular, since the seminal papers of McCulloch and Pitts \cite{MP43}, Hopfield \cite{H82}, Kauffman \cite{K69,K93} and Thomas \cite{T73,TA90}, they are omnipresent in the modeling of neural and gene networks (see \cite{B08,N15} for reviews). They are also essential tools in computer science, see \cite{ANLY00,GRF16,BGT14,CFG14a,GR15b} for instance. 

\medskip
The ``network'' terminology comes from the fact that the \EM{interaction graph} of $f$ is often considered as the main parameter of $f$: the vertex set is $[n]=\{1,\dots,n\}$ and there is an arc from $j$ to $i$ if $f_i$ depends on input $j$. The \EM{signed interaction graph} of $f$, denoted \EM{$G(f)$}, provides useful additional information about interactions, and is commonly considered in the context of gene networks: the vertex set is $[n]$ and there is a positive (negative) arc from $j$ to $i$ if there are $x,y\in \B^n$ that only differ in $x_j<y_j$ such that $f_i(y)-f_i(x)$ is positive (negative). Note that the presence of both a positive and a negative arc from one vertex to another is allowed.

\medskip
From a dynamical point of view, the successive iterations of $f$ describe the so called \EM{synchronous dynamics}: if $x^t$ is the configuration of the system at time $t$, then $x^{t+1}=f(x^t)$ is the configuration of the system at the next time. Hence, all components are updated in parallel at each time step. However, when BNs are used as models of natural systems, such as gene networks, synchronicity can be an issue. This led researchers to consider the \EM{(fully) asynchronous dynamics}, where one component is updated at each time step (see e.g. \cite{T91,TA90,TK01,A-J16}). If $x^t$ is the configuration of the system at time $t$, then the configuration at time $t+1$ is $x$ if $f(x)=x$ and, otherwise, a configuration $y$ obtained from $x$ by flipping a component $i$ such that $f_i(x)\neq x_i$. The asynchronous dynamics can be described by the paths of the \EM{asynchronous graph} of $f$, denoted \EM{$\Gamma(f)$}: the vertex set is $\B^n$, and there is an arc from $x$ to $y$ if and only if $y$ is obtained from $x$ by flipping a component $i$ such that $x_i\neq f_i(x)$. The asymptotic behaviors are described by the \EM{attractors} of $\Gamma(f)$, which are the inclusion-minimal \EM{trap sets}, where $X\subseteq \B^n$ is a trap set if $\Gamma(f)$ has no arc from a vertex in $X$ to a vertex outside $X$. In particular, we say that $\Gamma(f)$ is:
\begin{itemize}
\item
\EM{fixing} if all the attractors are of size one,
\item  
\EM{converging} if there is a unique attractor.
\end{itemize}

\medskip
In biological applications, and for gene networks in particular, the first reliable experimental information often concern the signed interaction graph while the actual dynamics are very difficult to observe \cite{TK01,N15}. One is thus faced with the following question: {\em what can be said about $\Gamma(f)$ according to $G(f)$?} An influential result is this direction is the following \cite{ADG04a,A08}: if $G(f)$ has no negative cycle, then $f$ has at least one fixed point; and if $G(f)$ has no positive cycle, then $f$ has at most one fixed point. Soon after, it was realized that ``fixed point'' can be replaced by ``asynchronous attractor'' giving: if $G(f)$ has no negative cycle, then $\Gamma(f)$ is fixing \cite{R10}; and if $G(f)$ has no positive cycle, then $\Gamma(f)$ is converging \cite{RC07}.  

\medskip
In this paper, we are interested in conditions on $G(f)$ that imply asymptotic properties in $\Gamma(f)$ which are weaker than the fixing and converging properties. To describe them, we need additional definitions. A \EM{subspace} is set of configurations $X$ such that, for some $I\subseteq [n]$ and $c:I\to \B$, we have $x\in X$ if and only if $x_i=c(i)$ for all $i\in I$. Hence a subspace is obtained by fixing some components. Given a set of configurations $X$, we denote by $[X]$ the smallest subspace containing $X$. A \EM{trap space} is a trap set which is also a subspace. We denote by $\langle X\rangle$ the smallest trap space containing $X$. Obviously, $[X]\subseteq \langle X\rangle$. We say that $\Gamma(f)$ is
\begin{itemize}
\item
\EM{separating} if $[A]\cap [B]=\emptyset$ for all distinct attractors $A,B$,
\item
\EM{trap-separating} if $\langle A\rangle\cap \langle B\rangle=\emptyset$  for all distinct attractors $A,B$, 
\item
\EM{trapping} if it is separating and $[A]=\langle A\rangle$ for each attractor $A$.
\end{itemize}
The trapping property has been introduced in \cite{NRT22} with the following equivalent definition: for each attractor $A$, $[A]=\langle A\rangle$ and $A$ is the unique attractor reachable from any state in $[A]$. 

\medskip

Arriving at a description of the attractor landscape is important for the identification of phenomena such as differentiation or stable periodicity, but is in general a hard problem, subject of ongoing research~\cite{klarner2015approximating,rozum2021parity}.
On the other hand, trap spaces can be computed more easily, for instance using logic programming~\cite{klarner2015computing}.
The number of minimal trap spaces provides a lower bound on the number of attractors. Moreover, an analysis of published biological models~\cite{klarner2015approximating} found that minimal trap spaces are often good approximations of attractors, meaning that each minimal trap space contains only one attractor (``univocality''), all attractors are found inside minimal trap spaces (``completeness''), and oscillating variables in attractors span the minimal containing trap space in all directions (``faithfullness'').
Under these conditions, model analyses that investigate reachability of attractors or existence of control strategies (see e.g.~\cite{cf2022control}) can be greatly facilitated.
Of particular interest are structural conditions on the interaction graph that can guarantee these properties.
In this work we look for conditions for the dynamics to be trap-separating, which implies that minimal trap spaces are complete and univocal, and for the dynamics to be trapping, which adds faithfulness of the minimal trap spaces, and investigate the ``worst case scenario'' where attractors cannot be separated even by subspaces.

\medskip
One easily check that fixing $\Rightarrow$ trapping $\Rightarrow$ trap-separating $\Rightarrow$ separating. Furthermore, converging $\Rightarrow$ trap-separating (but converging $\not\Rightarrow$ trapping). The situation is described at the top of \cref{fig:results}. We deduce that if $\Gamma(f)$ is not trap-separating, then it is neither converging nor fixing, and thus $G(f)$ has at least one positive cycle and at least one negative cycle. But can something stronger be said? This paper provides partial answers. 

\medskip
In particular, we prove that if $\Gamma(f)$ is not trap-separating, then $G(f)$ has a path from a negative cycle to a positive cycle. If $\Gamma(f)$ is non-separating, we say more:
\begin{itemize}
\item 
$G(f)$ has a positive cycle which intersects a negative cycle, and
\item
$G(f)$ has at least two negative cycles, and
\item
at least two vertices must be removed from $G(f)$ to destroy all the positive cycles. 
\end{itemize}
The first point is particularly interesting since little is known about the dynamical influence of such intersections (see however \cite{didier2012relations,MNRSS15,remy2016boolean,ARS17b,R18,mosse2020combinatorial}). Consider the following signed digraph, called $H_2$ (throughout the paper, green arcs are positive and red arcs are negative):
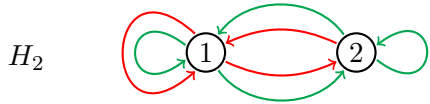
\[
H_2\qquad 
\begin{array}{c}
\begin{tikzpicture}
\useasboundingbox (-2.2,-0.7) rectangle (2.2,0.7);
\node[outer sep=1,inner sep=2,circle,draw,thick] (1) at ({180}:1){$1$};
\node[outer sep=1,inner sep=2,circle,draw,thick] (2) at ({0}:1){$2$};
\draw[Green,->,thick] (1.{180-20}) .. controls ({180-20}:2.3) and ({180+20}:2.3) .. (1.{180+20});
\draw[red,->,thick] (1.{180-60}) .. controls ({180-30}:2.8) and ({180+30}:2.8) .. (1.{180+60});
\draw[Green,->,thick] (2.{0-20}) .. controls ({0-20}:2.3) and ({0+20}:2.3) .. (2.{0+20});
\path[->,thick]
(1) edge[red,bend right=25] (2)
(1) edge[Green,bend right=55] (2)
(2) edge[red,bend right=25] (1)
(2) edge[Green,bend right=55] (1)
;
\end{tikzpicture}
\end{array}
\]
This is a minimal signed digraph which satisfies the three conditions given above, and we prove that if $\Gamma(f)$ is non-separating then, under some conditions on $G(f)$, the presence of $H_2$ is unavoidable. To be precise, let us say that a signed digraph $H$ with vertex set $V$ is \EM{embedded} in $G(f)$ if there is an injection $\phi:V\to [n]$ such that, for every positive (negative) arc of $H$ from $j$ to $i$, $G(f)$ contains a positive (negative) path from $\phi(j)$ to $\phi(i)$ whose internal vertices are not in $\phi(V)$. We prove that, if $\Gamma(f)$ is non-separating, then either $H_2$ is embedded in $G(f)$, or at least three vertices must be removed from $G(f)$ to destroy all the cycles. The dynamics associated to isolated complex motifs has been previously investigated, as well as relationships between feedback vertex numbers and number of attractors (see e.g. \cite{A08,didier2012relations,MNRSS15,remy2016boolean,ARS17,mosse2020combinatorial}); however, this is the first time, to our knowledge, that a connection between a dynamical property and the embedding of such a complex pattern is identified.

\medskip
A sufficient structural condition for the trapping property is identified in \cite{NRT22}: $\Gamma(f)$ is trapping if $G(f)$ has a \EM{linear cut}, that is, if it has no arc from a vertex of out-degree at least two to a vertex of in-degree at least two and every cycle contains a vertex of in- and out-degree one. Here, with the tools used to analyse signed digraphs that are non-separating, we provide a sufficient condition for $\Gamma(f)$ to be trapping which is rather different: we show that $\Gamma(f)$ is trapping when $G(f)$ is strong and has at most one negative cycle. We also provide in the same vein new sufficient conditions for $\Gamma(f)$ to be fixing or converging.

\medskip
We divide the results in statements that concern the intersection of positive and negative cycles (\cref{sec:intersection_positive_negative}), number of positive (\cref{sec:number-positive-cycles}) and negative cycles (\cref{sec:number-negative-cycles}), and graphs with feedback number two (\cref{sec:feedback-number-two}).
A summary of our results is given in \cref{fig:results}.
We conclude with some conjectures and open questions (\cref{sec:open-problems}).

\pgfdeclarelayer{background}
\pgfdeclarelayer{foreground}
\pgfsetlayers{background,main,foreground}
\def\base{4cm}
\def\b{2cm}
\def\h{1cm}
\begin{figure}
\resizebox{\textwidth}{!}{
\begin{tikzpicture}
    \node[rectangle,draw,minimum width=\base,minimum height=0.5*\base,line width=0.5mm] (fix) at (-0.1,0) {\begin{tabular}{c}\textbf{fixing}\\$|A|=1$ $\forall A$\end{tabular}};
    \node[rectangle,draw,minimum width=\base,minimum height=0.5*\base,line width=0.5mm] (conv) at (1.2*\base,0.3*\base) {\begin{tabular}{c}\textbf{converging}\\$\exists! A$\end{tabular}};
    \node[rectangle,draw,minimum width=\base,minimum height=0.5*\base,line width=0.5mm] (trap) at (2.4*\base,0) {\begin{tabular}{c}\textbf{trapping}\\separating,\\$\langle A\rangle=[A]$ $\forall A$\end{tabular}};
    \node[rectangle,draw,minimum width=1.2*\base,minimum height=0.82*\base,line width=0.5mm] (trapsep) at (3.7*\base,0.15*\base) {\begin{tabular}{c}\textbf{trap-separating}\\\hspace{2pt}\\$\langle A\rangle\cap \langle B\rangle=\emptyset$\\$\forall A\neq B$\end{tabular}};
    \node[rectangle,draw,minimum width=2.0*\base,minimum height=0.82*\base,line width=0.5mm] (sep) at (5.5*\base,0.15*\base) {\begin{tabular}{c}\textbf{separating}\\\hspace{2pt}\\$[A]\cap [B]=\emptyset$ $\forall A\neq B$\end{tabular}};

    \node[rectangle,draw,fill=black!10,minimum width=\b,minimum height=\h,line width=0.3mm,rounded corners=0.2cm] (nnc) at (-0.1,0-0.6*\base) {no negative cycle};
    \node[rectangle,draw,fill=black!10,minimum width=\b,minimum height=\h,line width=0.3mm,rounded corners=0.2cm] (npc) at (1.2*\base,-0.6*\base) {no positive cycle};
    \node[rectangle,draw,fill=black!10,minimum width=\b,minimum height=\h,line width=0.3mm,rounded corners=0.2cm] (lc) at (2.2*\base,0-0.6*\base) {linear cut};
    \node[rectangle,draw,fill=white,minimum width=\b,minimum height=\h,line width=0.3mm,rounded corners=0.2cm] (nopathnp) at (3.6*\base,0-0.6*\base) {\begin{tabular}{c}no path from neg.\\to pos. cycle\end{tabular}};
    \node[rectangle,draw,fill=white,minimum width=\b,minimum height=\h,line width=0.3mm,rounded corners=0.2cm] (oneneg) at (0.05*\base,0-1.16*\base) {\begin{tabular}{c}at least 1 pos. cycle,\\unique neg. cycle\\ meets all cycles\end{tabular}};
    \node[rectangle,draw,fill=white,minimum width=\b,minimum height=\h,line width=0.3mm,rounded corners=0.2cm] (stroneneg) at (0*\base,0-1.76*\base) {\begin{tabular}{c}strong, at least 1 pos.\\cycle, unique neg. cycle\\meets all cycles\end{tabular}};
    \node[rectangle,draw,fill=white,minimum width=\b,minimum height=\h,line width=0.3mm,rounded corners=0.2cm] (onepos1) at (1.1*\base,0-1.16*\base) {\begin{tabular}{c}at least 1 neg. cycle,\\unique pos. cycle\\meets all cycles\end{tabular}};
    \node[rectangle,draw,fill=white,minimum width=\b,minimum height=\h,line width=0.3mm,rounded corners=0.2cm] (stronepos) at (1.2*\base,0-1.76*\base) {\begin{tabular}{c}strong, at least 1 neg.\\cycle, unique pos. cycle\\meets all cycles\end{tabular}};
    \node[rectangle,draw,fill=white,minimum width=\b,minimum height=\h,line width=0.3mm,rounded corners=0.2cm] (strmaxoneneg) at (2.4*\base,0-1.22*\base) {\begin{tabular}{c}strong and at most\\one neg. cycle\end{tabular}};
    \node[rectangle,draw,fill=white,minimum width=\b,minimum height=\h,line width=0.3mm,rounded corners=0.2cm] (onepos) at (3.6*\base,0-1.1*\base) {\begin{tabular}{c}unique pos. cycle\\meets all cycles\end{tabular}};
    \node[rectangle,draw,fill=white,minimum width=\b,minimum height=\h,line width=0.3mm,rounded corners=0.2cm] (no) at (4.8*\base,0-0.6*\base) {\begin{tabular}{c}no neg. and pos.\\cycles intersect\end{tabular}};
    \node[rectangle,draw,fill=white,minimum width=\b,minimum height=\h,line width=0.3mm,rounded corners=0.2cm] (fvntwo) at (5.9*\base,0-0.6*\base) {\begin{tabular}{c}feedback number 2 and\\no embedding of $H_2$\end{tabular}};
    \node[rectangle,draw,fill=white,minimum width=\b,minimum height=\h,line width=0.3mm,rounded corners=0.2cm] (pfvn) at (4.8*\base,0-1.1*\base) {\begin{tabular}{c}pos. feedback\\number 1\end{tabular}};
    \node[rectangle,draw,fill=white,minimum width=\b,minimum height=\h,line width=0.3mm,rounded corners=0.2cm] (nfvn) at (6.1*\base,0-1.1*\base) {\begin{tabular}{c}neg. feedback\\number 1\end{tabular}};
    \node[rectangle,draw,fill=white,minimum width=\b,minimum height=\h,line width=0.3mm,rounded corners=0.2cm] (atmostonep) at (4.8*\base,0-1.6*\base) {\begin{tabular}{c}at most one\\pos. cycle\end{tabular}};
    \node[rectangle,draw,fill=white,minimum width=\b,minimum height=\h,line width=0.3mm,rounded corners=0.2cm] (atmostonen) at (5.8*\base,0-1.6*\base) {\begin{tabular}{c}at most one\\neg. cycle\end{tabular}};

    \draw[->,ultra thick] ([shift={(0cm,-0.3cm)}]fix.east) -- ([shift={(0cm,-0.3cm)}]trap.west);
    \draw[->,ultra thick] ([shift={(0cm,0.3cm)}]conv.east) -- ([shift={(0cm,0.9cm)}]trapsep.west);
    \draw[->,ultra thick,red,dotted] ([shift={(0cm,-0.7cm)}]conv.east) -- ([shift={(0cm,0.5cm)}]trap.west) node[pos=0.5,above] {ex.\ref{ex:non-trapping}};
    \draw[->,ultra thick] (trap) -- ([shift={(0cm,-0.58cm)}]trapsep.west);
    \draw[->,ultra thick] (trapsep) -- (sep);

    \draw[->,ultra thick,gray] (nnc) -- (fix) node[midway,right] {\cite{R10}};
    \begin{pgfonlayer}{background}
        \draw[->,ultra thick,gray] (npc) -- (conv) node[pos=0.2,right,on background layer] {\cite{RC07}};
    \end{pgfonlayer}
    \draw[->,ultra thick,gray] (lc) -- ([shift={(-0.8cm,-1cm)}]trap.center) node[midway,right] {\cite{NRT22}};

    \draw[->,ultra thick] (nopathnp) -- ([shift={(-0.4cm,-1.65cm)}]trapsep.center) node[midway,right] {Thm.\ref{thm:trap-sep}};
    \draw[->,ultra thick] (no) -- ([shift={(-2.8cm,-1.65cm)}]sep.center) node[midway,right] {Thm.\ref{thm:sep}};
    \draw[->,ultra thick] (fvntwo) -- ([shift={(1.6cm,-1.65cm)}]sep.center) node[midway,right] {Thm.\ref{thm:fvs2}};
    \draw[->,ultra thick] ([shift={(0.0cm,0.6cm)}]atmostonep.center) -- ([shift={(0.0cm,-0.6cm)}]pfvn.center);
    \draw[->,ultra thick] ([shift={(1.0cm,0.6cm)}]atmostonen.center) -- ([shift={(-0.2cm,-0.6cm)}]nfvn.center);
    \begin{pgfonlayer}{background}
        \draw[->,ultra thick] ([shift={(-1.2cm,0cm)}]strmaxoneneg.north) -- ([shift={(-1.2cm,0cm)}]trap.south) node[pos=0.27,right] {Thm.\ref{thm:NEGATIVE}};
        \draw[->,ultra thick] ([shift={(-2.0cm,0.6cm)}]stroneneg.center) -- ([shift={(-1.9cm,-1.0cm)}]fix.center) node[pos=0.1,right] {Prop.\ref{pro:fixing}};
        \draw[->,ultra thick] ([shift={(+1.8cm,0cm)}]stronepos.north) -- ([shift={(+1.8cm,0cm)}]conv.south) node[pos=0.05,right] {Prop.\ref{pro:one_positive_cycle}};
        \draw[->,ultra thick] ([shift={(-0.4cm,0.6cm)}]onepos.center) -- ([shift={(-0.8cm,-1.65cm)}]trapsep.center) node[pos=0.15,right] {Prop.\ref{pro:one_positive_cycle}};
        \draw[->,ultra thick] ([shift={(-0.6cm,0.6cm)}]pfvn.center) -- ([shift={(-3.4cm,-1.65cm)}]sep.center) node[pos=0.15,right] {Thm.\ref{thm:PFN}};
        \draw[->,ultra thick] ([shift={(-0.6cm,0.6cm)}]atmostonen.center) -- ([shift={(0.6cm,-1.65cm)}]sep.center) node[pos=0.08,right] {Thm.\ref{thm:NEGATIVE}};
    \end{pgfonlayer}

    \draw[->,ultra thick,red,dotted] ([shift={(1.4cm,0cm)}]oneneg.north) -- +(0cm,0.6cm) -- +(0.55cm,0.6cm) -- +(0.55cm,2.3cm) node[pos=0.4,right] {ex.\ref{ex:sep-not-conv-fix-graph}} -- +(0cm,2.3cm) -- ([shift={(1.7cm,0cm)}]fix.south);
    \draw[->,ultra thick,red,dotted] ([shift={(-1.4cm,0cm)}]onepos1.north) -- +(0cm,0.6cm) -- +(-0.85cm,0.6cm) -- +(-0.85cm,2.3cm) -- +(0.05cm,2.3cm) -- ([shift={(-1.75cm,0cm)}]conv.south);
    \draw[->,ultra thick,red,dotted] (nopathnp) -| node[pos=0.87,right] {ex.\ref{ex:non-trapping}} ([shift={(1.6cm,-1.0cm)}]trap.center);
    \draw[->,ultra thick,red,dotted] ([shift={(0cm,0.3cm)}]onepos.west) -| node[pos=0.27,above] {ex.\ref{ex:pos-not-trapping}} ([shift={(1.4cm,-1.0cm)}]trap.center);
    \draw[->,ultra thick,red,dotted] (no) -| node[pos=0.87,right] {ex.\ref{ex:not-trap-sep}} ([shift={(2.0cm,-1.65cm)}]trapsep.center);
    \draw[->,ultra thick,red,dotted] (pfvn) -| ([shift={(2.0cm,-1.65cm)}]trapsep.center);
    \draw[->,ultra thick,red,dotted] (atmostonep) -| ([shift={(2.0cm,-1.65cm)}]trapsep.center);
    \draw[-,ultra thick,red,dotted] (atmostonen) -- (atmostonep);
    \draw[-,ultra thick,red,dotted] (fvntwo) -- (no);
    \draw[->,ultra thick,red,dotted] (nfvn) -| node[pos=0.55,right] {ex.\ref{ex:negative_feedback_1}} ([shift={(-0.85cm,-1.65cm)}]sep.center);
\end{tikzpicture}
}\caption{\label{fig:results} Summary of the main definitions and results of this work, and some known results (indicated with gray boxes and arrows). $A$ and $B$ stand for attractors of asynchronous dynamics. Counterexamples are indicated in dashed red arrows.}
\end{figure}
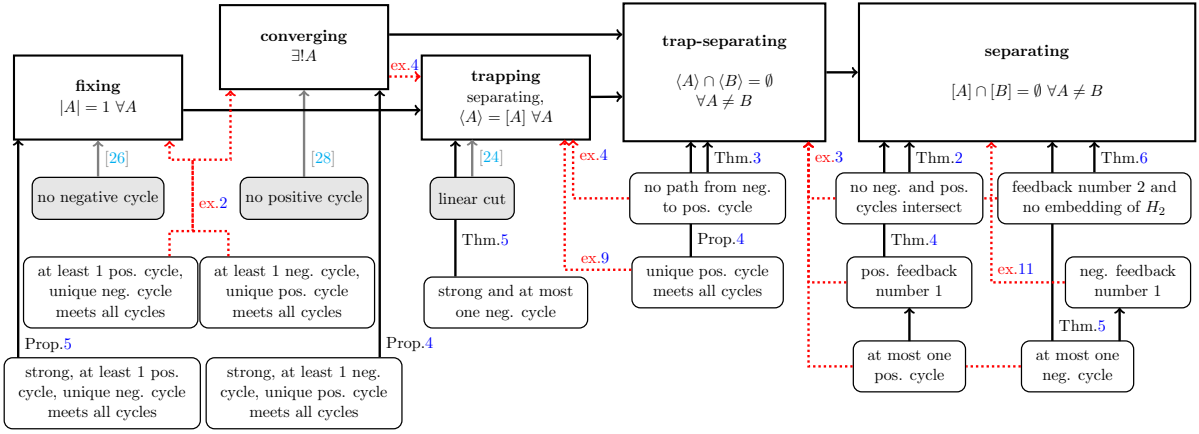

\section{Definitions and background}

\subsection{Digraphs and signed digraphs}

A \EM{digraph} is a pair $G=(V,E)$ where $V$ is a set of vertices and $E\subseteq V^2$ is a set of arcs. Given $I\subseteq V$, the subgraph of $G$ induced by $I$ is denoted $G[I]$, and $G\setminus I$ means $G[V\setminus I]$. A \EM{strongly connected component} (\EM{strong component} for short) of a digraph $G$ is an induced subgraph which is \EM{strongly connected} (\EM{strong} for short) and maximal for this property. A strong component $G[I]$ is \EM{initial} if $G$ has no arc from $V\setminus I$ to $I$, and \EM{terminal} if $G$ has no arc from $I$ to $V\setminus I$.  A digraph is \EM{trivial} if it has a unique vertex and no arc. 

\medskip
A \EM{signed digraph} $G$ is a pair $(V,E)$ where $E\subseteq V^2\times\{-1,1\}$. If $(j,i,s)\in E$ then $G$ has an arc from $j$ to $i$ of sign $s$; we also say that $j$ is an in-neighbor of $i$ of sign $s$ and that $i$ is an out-neighbor of $j$ of sign $s$. We say that $G$ is \EM{simple} if it does not have both a positive arc and a negative arc from one vertex to another, and \EM{full-positive} if all its arcs are positive. A subgraph of $G$ is a signed digraph $(V',E')$ with $V'\subseteq V$ and $E'\subseteq E$. Cycles and paths of $G$ are regarded as simple subgraphs. The \EM{sign of a cycle or a path} of $G$ is the product of the signs of its arcs. We say that $I$ is a \EM{feedback vertex set} if $G\setminus I$ has no cycle. The \EM{feedback number} of $G$ is the minimum size of a \EM{feedback vertex set} of $G$. Similarly, we say that $I$ is a \EM{positive (negative) feedback vertex set} if $G\setminus I$ has no positive (negative) cycle. The \EM{positive (negative) feedback number} of $G$ is the minimum size of a \EM{positive (negative) feedback vertex set} of $G$. We say that $G$ has a \EM{linear cut} if it has no arc from a vertex of out-degree at least two to a vertex of in-degree at least two and every cycle contains a vertex of in- and out-degree one. The underlying (unsigned) digraph of $G$ has vertex set $V$ and an arc from $j$ to $i$ if $G$ has a positive or a negative arc from $j$ to $i$. Every graph concept made on $G$ that does not involved signs are tacitly made on its underlying digraph. For instance, $G$ is strongly connected if its underlying digraph is. In the following, $G$ always denotes a signed digraph with vertex set $V$.

\subsection{Configurations}

The set of maps from $V$ to $\B$ is denoted \EM{$\B^V$} and called set of \EM{configurations} on $V$. Given such a configuration $x$ and $i\in V$, we denote \EM{$x_i$} the image of $i$ by $x$ and for $I\subseteq V$, \EM{$x_I$} is the restriction of $x$ to $I$. In examples, we always have $V=[n]$ for some $n\geq 1$ and we identify $x$ and the binary sequence $x_1x_2\dots x_n$. We denote by \EM{$e_I$} the configuration such that $(e_I)_i=1$ if $i\in I$ and $(e_I)_i=0$ otherwise. We write \EM{$e_i$} instead of $e_{\{i\}}$. Let $x,y$ be two configurations on $V$. We denote by \EM{$x+y$} the configuration $z$ on $V$ with $z_i=x_i+y_i$ for all $i\in V$, where the addition is modulo $2$, and by $\bar{x}$ the configuration $x+e_V$. We denote by \EM{$\Delta(x,y)$} the set of $i\in V$ with $x_i\neq y_i$. The \EM{Hamming distance} between $x$ and $y$ is $\EMM{d(x,y)}=|\Delta(x,y)|$. We equip $\B^V$ with the partial order \EM{$\leq $} defined by $x\leq y$ if and only if $x_i\leq y_i$ for $i\in V$. We denote by \EM{$\ZERO$} and \EM{$\ONE$} the all-zero and all-one configurations, that is, the minimal and maximal element of $\B^V$. If $x\leq y$ then \EM{$[x,y]$} is the set of configurations $z$ on $V$ such that $x\leq z\leq y$. Let $X\subseteq \B^V$. We denote by \EM{$\Delta(X)$} the set of $i\in V$ such that $x_i\neq y_i$ for some $x,y\in X$. We say that $X$ is a \EM{subspace} if $X=[x,y]$ for some configurations $x,y$ on $V$. We denote by \EM{$[X]$} the smallest subspace of $\B^V$ containing $X$. 

\subsection{Boolean networks}

A \EM{Boolean network} (BN) with component set $V$ is a map $f:\B^V\to\B^V$. We denote by $F(V)$ the set of BNs with component set $V$ and for $n\geq 1$ we write $F(n)$ instead of $F([n])$. We say that $f$ is \EM{monotone} if $x\leq y$ implies $f(x)\leq f(y)$ for all  configurations $x,y$ on $V$. We denote by $G(f)$ the \EM{signed interaction graph} of $f$: it is the signed digraph with vertex set $V$ such that, for all $i,j\in V$, there is a positive (negative) arc from $j$ to $i$ if there exist a configuration $x$ on $V$ such that $x_j=0$ and $f_i(x+e_j)-f_i(x)$ is positive (negative); we can have both a positive and a negative arc from one component to another. Given a signed digraph $G$, a BN \EM{on} $G$ is a BN with interaction graph equal to $G$. We denote by $F(G)$ the set of BNs on $G$. Let $x,y$ be two configurations on $V$ with $x\leq y$. The \EM{subnetwork} of $f$ induced by $[x,y]$ is the BN $h$ with component set $I=\Delta(x,y)$ defined by $h(z_I)=f(z)_I$ for all $z\in [x,y]$. Intuitively, $h$ is obtained from $f$ by fixing to $x_i=y_i$ each component $i\in V\setminus I$. One can easily check that $G(h)$ is a subgraph of $G(f)[I]$.

\subsection{Asynchronous graphs}

An \EM{asynchronous graph} $\Gamma$  on $\B^V$ is a digraph with vertex set $\B^V$ such that, for every arc $x\to y$, there is exactly one $i\in V$ such that $x_i\neq y_i$; this component $i$ is the \EM{direction} of the arc. We say that $x\to y$ is \EM{increasing} if $x\leq y$ and \EM{decreasing} if $y\leq x$. Let $X$ be a set of configurations on $V$. We say that an arc $x\to y$ of $\Gamma$ \EM{leaves} $X$ if $x\in X$ and $y\not\in X$. We say that $X$ is a \EM{trap set} of $\Gamma$ if no arc leaves $X$. A \EM{trap space} of $\Gamma$ is a subspace which is also a trap set. We denote by \EM{$\langle X\rangle$} the smallest trap space containing $X$ (which exists since $\B^V$ is a trap space). An \EM{attractor} of $\Gamma$ is a terminal strong component of $\Gamma$ or, equivalently, an inclusion-minimal non-empty trap set (there is at least one attractor since $\B^V$ is a trap set). Given a BN $f\in F(V)$, we denote by $\Gamma(f)$ the \EM{asynchronous graph} of $f$, that is, the asynchronous graph on $\B^V$ with an arc $x\to y$ in the direction $i$ if and only if $f_i(x)\neq x_i$.
It is easy to see that each asynchronous graph on $\B^V$ is the asynchronous graph of a unique BN with component set $V$. For a subspace $X$ of $\{0,1\}^V$, we denote by $\Gamma(f)[X]$ the subgraph of $\Gamma(f)$ induced by $X$. Clearly, this is the asynchronous dynamics of the subnetwork of $f$ induced by $X$.

\medskip
An asynchronous graph is:
\begin{itemize}
\item
\EM{fixing} if all the attractors are of size one,
\item  
\EM{converging} if there is a unique attractor,
\item
\EM{separating} if $[A]\cap [B]=\emptyset$ for all distinct attractors $A,B$,
\item
\EM{trap-separating} if $\langle A\rangle\cap \langle B\rangle=\emptyset$  for all distinct attractors $A,B$, 
\item
\EM{trapping} if it is separating and $[A]=\langle A\rangle$ for all attractor $A$.
\end{itemize}
We abusively say that a signed digraph $G$ is \EM{converging} (resp. \EM{fixing}, \EM{trapping}, \EM{trap-separating}, \EM{separating}) if the asynchronous graph of every $f\in F(G)$ is converging (resp. fixing, trapping, trap-separating, separating). One easily check that fixing $\Rightarrow$ trapping $\Rightarrow$ trap-separating $\Rightarrow$ separating. Furthermore, converging $\Rightarrow$ trap-separating, but converging $\not\Rightarrow$ trapping, see \cref{ex:non-trapping}. Here are some sufficient conditions for $G$ to be fixing, converging or trapping.

\begin{theorem}\label{thm:bib}
For every signed digraph $G$, 
\begin{itemize}
  \item if $G$ has no cycle, then $G$ is converging and fixing \cite{R95};
  \item if $G$ has no positive cycle, then $G$ is converging \cite{RC07};
  \item if $G$ has no negative cycle, then $G$ is fixing \cite{R10};
  \item if $G$ has a linear cut, then $G$ is trapping \cite{NRT22}.
\end{itemize}
\end{theorem}

For some proofs in this paper we will rely on the following results, which imply the second and third points of the previous theorem. 

\begin{lemma}[\cite{A08,RC07}]\label{lem:A08}
Let $f\in F(G)$ and let $x,y$ be two configurations on $V$ such that $f_i(x)=x_i$ and $f_i(y)=y_i$ for all $i\in\Delta(x,y)$. Then $G[\Delta(x,y)]$ has a positive cycle $C$, and if $C$ contains an arc from $j$ to $i$ then the sign of this arc is $(y_j-x_j)(y_i-x_i)$.
\end{lemma}

\begin{lemma}[\cite{R10}]\label{lem:R10}
Let $f\in F(G)$ and suppose that $\Gamma(f)$ has an attractor $A$ of size at least two. Then $G[\Delta(A)]$ has a negative cycle.
\end{lemma}

By combining the two theorems we can already state a result about non-intersecting positive and negative cycles in trap-separating networks that have at least two attractors and at least one cyclic attractor.

\begin{proposition}\label{pro:disjoint-opposite}
Let $f\in F(G)$. If $\Gamma(f)$ is trap-separating but neither converging nor fixing, then $G$ has vertex-disjoint cycles of distinct sign.
\end{proposition}

\begin{proof}
Suppose that $\Gamma(f)$ is trap-separating but neither converging nor fixing. Then it has two attractors $A,B$ with $|A|>1$ and disjoint trap spaces $X,Y$ with $A\subseteq X$ and $B\subseteq Y$. Consider configurations $x \in X$ and $y \in Y$ which minimize $d(x,y)$. Then $x$ and $y$ are fixed points for the subnetwork of $f$ induced by $[x,y]$. Hence $G[\Delta(x,y)]$ has a positive cycle by \cref{lem:A08}. By \cref{lem:R10}, $G[\Delta(A)]$ has a negative cycle. Since $\Delta(x,y)\cap \Delta(A)=\emptyset$ this proves the proposition.
\end{proof}

The conclusion of the proposition does not hold if we replace trap-separating with separating:
the following presents an example of asynchronous graph that is separating but not converging nor fixing, while its corresponding signed interaction graph does not have disjoint positive or negative cycles.

\begin{example}\label{ex:sep-not-conv-fix}
The BN $f\in F(5)$ defined by $f_1(x)=x_4 x_5 \lor \bar x_4 \bar x_5$, $f_2(x)=x_1 \bar x_5 \lor x_5 \bar x_1$, $f_3(x)=x_2 \bar x_5 \lor x_5 \bar x_2$, $f_4(x)=x_3 \bar x_5 \lor x_5 \bar x_3$, $f_5(x)=x_1 x_3 \bar x_2 \lor x_1 x_4 \bar x_3 \lor x_2 \bar x_1 \bar x_3 \lor x_3 \bar x_1 \bar x_4$ has two cyclic attractors $A$ and $B$, with $[A]=\{x_5=0\}$ and $[B]=\{x_5=1\}$.
In addition, $\Gamma(f)$ has an arc from $00001$ to $00000$ and an arc from $10100$ to $10101$.
So $\Gamma(f)$ is separating but not converging, not fixing and not trap-separating. $G(f)$ does not have a positive and a negative cycle that are disjoint.
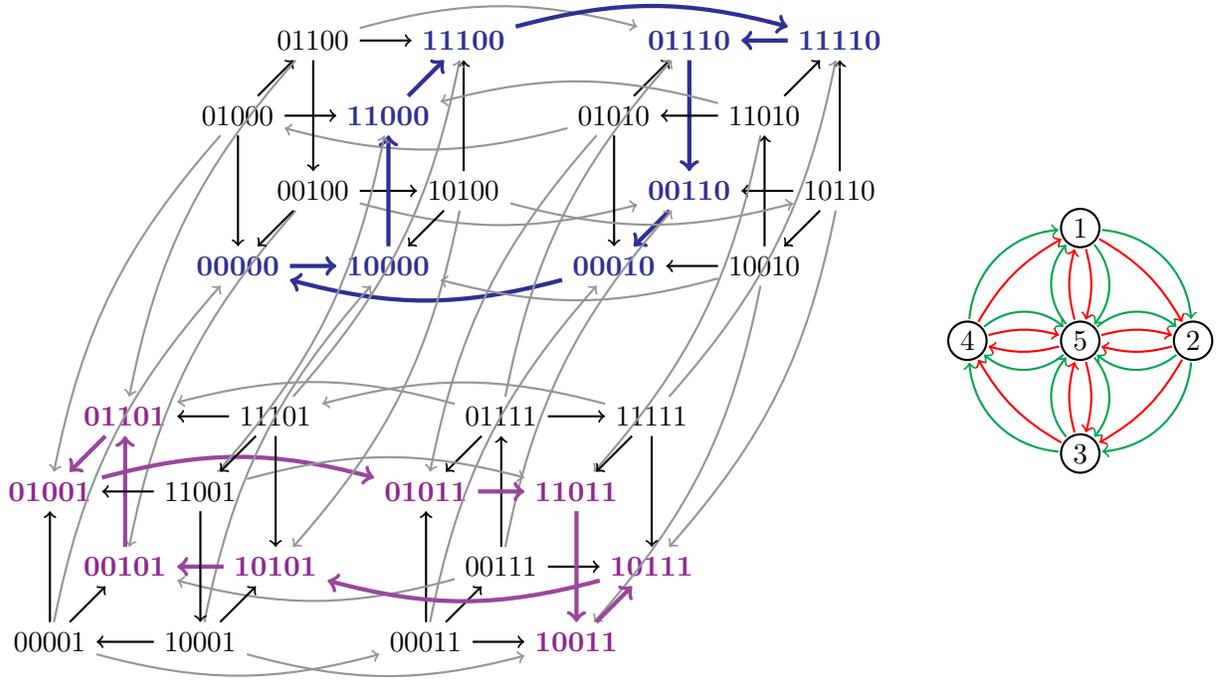
\[
\begin{array}{c}
\begin{tikzpicture}
\pgfmathparse{1}
\node (00000) at (0,0){{\boldmath \textcolor{Blue}{$00000$} \unboldmath}};
\node (00100) at (1,1){$00100$};
\node (01000) at (0,2){$01000$};
\node (01100) at (1,3){$01100$};
\node (10000) at (2,0){{\boldmath \textcolor{Blue}{$10000$} \unboldmath}};
\node (10100) at (3,1){$10100$};
\node (11000) at (2,2){{\boldmath \textcolor{Blue}{$11000$} \unboldmath}};
\node (11100) at (3,3){{\boldmath \textcolor{Blue}{$11100$} \unboldmath}};
\node (00010) at (5,0){{\boldmath \textcolor{Blue}{$00010$} \unboldmath}};
\node (00110) at (6,1){{\boldmath \textcolor{Blue}{$00110$} \unboldmath}};
\node (01010) at (5,2){$01010$};
\node (01110) at (6,3){{\boldmath \textcolor{Blue}{$01110$} \unboldmath}};
\node (10010) at (7,0){$10010$};
\node (10110) at (8,1){$10110$};
\node (11010) at (7,2){$11010$};
\node (11110) at (8,3){{\boldmath \textcolor{Blue}{$11110$} \unboldmath}};
\node (00001) at (0-2.5,0-5){$00001$};
\node (00101) at (1-2.5,1-5){{\boldmath \textcolor{Plum}{$00101$} \unboldmath}};
\node (01001) at (0-2.5,2-5){{\boldmath \textcolor{Plum}{$01001$} \unboldmath}};
\node (01101) at (1-2.5,3-5){{\boldmath \textcolor{Plum}{$01101$} \unboldmath}};
\node (10001) at (2-2.5,0-5){$10001$};
\node (10101) at (3-2.5,1-5){{\boldmath \textcolor{Plum}{$10101$} \unboldmath}};
\node (11001) at (2-2.5,2-5){$11001$};
\node (11101) at (3-2.5,3-5){$11101$};
\node (00011) at (5-2.5,0-5){$00011$};
\node (00111) at (6-2.5,1-5){$00111$};
\node (01011) at (5-2.5,2-5){{\boldmath \textcolor{Plum}{$01011$} \unboldmath}};
\node (01111) at (6-2.5,3-5){$01111$};
\node (10011) at (7-2.5,0-5){{\boldmath \textcolor{Plum}{$10011$} \unboldmath}};
\node (10111) at (8-2.5,1-5){{\boldmath \textcolor{Plum}{$10111$} \unboldmath}};
\node (11011) at (7-2.5,2-5){{\boldmath \textcolor{Plum}{$11011$} \unboldmath}};
\node (11111) at (8-2.5,3-5){$11111$};
\path[thick,->,draw,black]
(00000) edge[ultra thick,Blue] (10000)
(00001) edge[Gray,bend right=15] (00011)
(00001) edge (01001)
(00001) edge (00101)
(00001) edge[Gray,bend left=15] (00000)
(00010) edge[ultra thick,Blue,bend left=15] (00000)
(00011) edge (00111)
(00011) edge (01011)
(00011) edge (10011)
(00011) edge[Gray,bend left=15] (00010)
(00100) edge (10100)
(00100) edge[Gray,bend right=15] (00101)
(00100) edge[Gray,bend right=15] (00110)
(00100) edge (00000)
(00101) edge[ultra thick,Purple] (01101)
(00110) edge[ultra thick,Blue] (00010)
(00111) edge (10111)
(00111) edge[Gray,bend left=15] (00101)
(00111) edge[Gray,bend left=15] (00110)
(00111) edge (01111)
(01000) edge[Gray,bend right=15] (01001)
(01000) edge (01100)
(01000) edge (11000)
(01000) edge (00000)
(01001) edge[ultra thick,Purple,bend left=15] (01011)
(01010) edge[Gray,bend left=15] (01000)
(01010) edge[Gray,bend right=15] (01011)
(01010) edge (01110)
(01010) edge (00010)
(01011) edge[ultra thick,Purple] (11011)
(01100) edge (00100)
(01100) edge[Gray,bend right=15] (01101)
(01100) edge[Gray,bend left=15] (01110)
(01100) edge (11100)
(01101) edge[ultra thick,Purple] (01001)
(01110) edge[ultra thick,Blue] (00110)
(01111) edge (11111)
(01111) edge[Gray,bend right=15] (01101)
(01111) edge (01011)
(01111) edge[Gray,bend left=15] (01110)
(10000) edge[ultra thick,Blue] (11000)
(10001) edge[Gray,bend left=15] (10000)
(10001) edge (10101)
(10001) edge[Gray,bend right=15] (10011)
(10001) edge (00001)
(10010) edge[Gray,bend left=15] (10000)
(10010) edge (11010)
(10010) edge[Gray,bend left=15] (10011)
(10010) edge (00010)
(10011) edge[ultra thick,Purple] (10111)
(10100) edge (10000)
(10100) edge[Gray,bend right=15] (10110)
(10100) edge (11100)
(10100) edge[Gray,bend left=15] (10101)
(10101) edge[ultra thick,Purple] (00101)
(10110) edge (11110)
(10110) edge[Gray,bend left=15] (10111)
(10110) edge (10010)
(10110) edge (00110)
(10111) edge[ultra thick,Purple,bend left=15] (10101)
(11000) edge[ultra thick,Blue] (11100)
(11001) edge[Gray,bend left=15] (11011)
(11001) edge (01001)
(11001) edge (10001)
(11001) edge[Gray,bend right=15] (11000)
(11010) edge (11110)
(11010) edge[Gray,bend right=15] (11000)
(11010) edge (01010)
(11010) edge[Gray,bend left=15] (11011)
(11011) edge[ultra thick,Purple] (10011)
(11100) edge[ultra thick,Blue,bend left=15] (11110)
(11101) edge (10101)
(11101) edge (01101)
(11101) edge (11001)
(11101) edge[Gray,bend right=15] (11100)
(11110) edge[ultra thick,Blue] (01110)
(11111) edge (11011)
(11111) edge (10111)
(11111) edge[Gray,bend right=15] (11110)
(11111) edge[Gray,bend right=15] (11101)
;
\end{tikzpicture}
\end{array}
\quad
\begin{array}{c}
\begin{tikzpicture}
\node[outer sep=1,inner sep=2,circle,draw,thick] (1) at (90:1.5){$1$};
\node[outer sep=1,inner sep=2,circle,draw,thick] (2) at ({0}:1.5){$2$};
\node[outer sep=1,inner sep=2,circle,draw,thick] (4) at ({180}:1.5){$4$};
\node[outer sep=1,inner sep=2,circle,draw,thick] (3) at (-90:1.5){$3$};
\node[outer sep=1,inner sep=2,circle,draw,thick] (5) at (0:0){$5$};
\path[->,thick]
(1) edge[Green,bend left=40] (2)
(1) edge[red,bend left=15] (2)
(2) edge[Green,bend left=40] (3)
(2) edge[red,bend left=15] (3)
(3) edge[Green,bend left=40] (4)
(3) edge[red,bend left=15] (4)
(4) edge[Green,bend left=40] (1)
(4) edge[red,bend left=15] (1)
(1) edge[Green,bend left=40] (5)
(1) edge[red,bend left=15] (5)
(2) edge[Green,bend left=40] (5)
(2) edge[red,bend left=15] (5)
(3) edge[Green,bend left=40] (5)
(3) edge[red,bend left=15] (5)
(4) edge[Green,bend left=40] (5)
(4) edge[red,bend left=15] (5)
(5) edge[Green,bend left=40] (1)
(5) edge[red,bend left=15] (1)
(5) edge[Green,bend left=40] (2)
(5) edge[red,bend left=15] (2)
(5) edge[Green,bend left=40] (3)
(5) edge[red,bend left=15] (3)
(5) edge[Green,bend left=40] (4)
(5) edge[red,bend left=15] (4)
;
\end{tikzpicture}
\end{array}
\]
\end{example}

If we are interested in graphs that are separating but not converging nor fixing and do not have disjoint positive or negative cycles, then we can identify examples in smaller dimension.

\begin{example}\label{ex:sep-not-conv-fix-graph}
Consider the following signed digraph $G$, which has exactly one negative and one positive cycle that intersect:
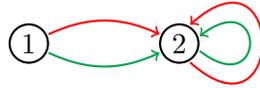
\[
\begin{tikzpicture}
\node[outer sep=1,inner sep=2,circle,draw,thick] (1) at ({180}:1){$1$};
\node[outer sep=1,inner sep=2,circle,draw,thick] (2) at ({0}:1){$2$};
\draw[red,->,thick] (2.{-60}) .. controls ({0-30}:2.8) and ({0+30}:2.8) .. (2.{0+60});
\draw[Green,->,thick] (2.{0-20}) .. controls ({0-20}:2.3) and ({0+20}:2.3) .. (2.{0+20});
\path[->,thick]
(1) edge[red,bend left=25] (2)
(1) edge[Green,bend right=25] (2)
;
\end{tikzpicture}
\]
It is neither converging nor fixing, but is separating, since all BNs in $F(G)$ are either converging or fixing.
Indeed, the asynchronous graphs of the four BNs in $F(G)$ are as follows:
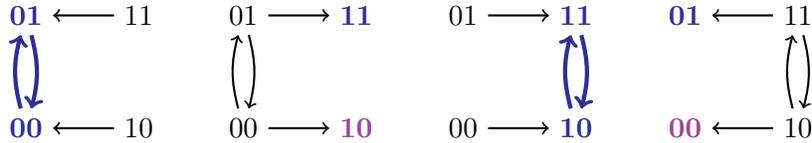
\[
\begin{array}{c}
\begin{tikzpicture}
\pgfmathparse{1}
\node (00) at (0,0){{\boldmath \textcolor{Blue}{$00$} \unboldmath}};
\node (01) at (0,1.5){{\boldmath \textcolor{Blue}{$01$} \unboldmath}};
\node (10) at (1.5,0){$10$};
\node (11) at (1.5,1.5){$11$};
\path[thick,->,draw,black]
(00) edge[bend left=15,Blue,ultra thick] (01)
(11) edge (01)
(10) edge (00)
(01) edge[bend left=15,Blue,ultra thick] (00)
;
\end{tikzpicture}
\end{array}
\quad
\begin{array}{c}
\begin{tikzpicture}
\pgfmathparse{1}
\node (00) at (0,0){$00$};
\node (01) at (0,1.5){$01$};
\node (10) at (1.5,0){{\boldmath \textcolor{Purple}{$10$} \unboldmath}};
\node (11) at (1.5,1.5){{\boldmath \textcolor{Blue}{$11$} \unboldmath}};
\path[thick,->,draw,black]
(00) edge[bend left=15] (01)
(01) edge (11)
(00) edge (10)
(01) edge[bend left=15] (00)
;
\end{tikzpicture}
\end{array}
\quad
\begin{array}{c}
\begin{tikzpicture}
\pgfmathparse{1}
\node (00) at (0,0){$00$};
\node (01) at (0,1.5){$01$};
\node (10) at (1.5,0){{\boldmath \textcolor{Blue}{$10$} \unboldmath}};
\node (11) at (1.5,1.5){{\boldmath \textcolor{Blue}{$11$} \unboldmath}};
\path[thick,->,draw,black]
(11) edge[bend left=15,Blue,ultra thick] (10)
(01) edge (11)
(00) edge (10)
(10) edge[bend left=15,Blue,ultra thick] (11)
;
\end{tikzpicture}
\end{array}
\quad
\begin{array}{c}
\begin{tikzpicture}
\pgfmathparse{1}
\node (00) at (0,0){{\boldmath \textcolor{Purple}{$00$} \unboldmath}};
\node (01) at (0,1.5){{\boldmath \textcolor{Blue}{$01$} \unboldmath}};
\node (10) at (1.5,0){$10$};
\node (11) at (1.5,1.5){$11$};
\path[thick,->,draw,black]
(11) edge[bend left=15] (10)
(11) edge (01)
(10) edge (00)
(10) edge[bend left=15] (11)
;
\end{tikzpicture}
\end{array}
\]
\end{example}

\subsection{Switches}

An \EM{isometry} on $\B^V$ is a bijection $\pi\colon\B^V\to\B^V$ such that $d(x,y)=d(\pi(x),\pi(y))$ for all $x,y\in\B^V$. Let $\Gamma,\Gamma'$ be two asynchronous graphs on $\B^V$. We say that $\Gamma$ and $\Gamma'$ are isometric if there is an isometry $\pi$ on $\B^V$ such that $x\to y$ is an arc of $\Gamma'$ if and only if $\pi(x)\to \pi(y)$ is an arc of $\Gamma'$. One can easily check that $\Gamma$ and $\Gamma'$ are isometric, then $\Gamma$ is \EM{converging} (resp. \EM{fixing}, \EM{trapping}, \EM{trap-separating}, \EM{separating}) if and only if $\Gamma'$ is converging (resp. fixing, trapping, trap-separating, separating).

Let $I\subseteq V$ and, for all $i\in V$, let $\sigma_I(i)=1$ if $i\in I$ and $\sigma_I(i)=-1$ if otherwise. The \EM{$I$-switch} of $G$ is the signed digraph $\EMM{G^I}=(V,E^I)$ with $\EMM{E^I}=\{(j,i,\sigma_I(j)\cdot s\cdot\sigma_I(i))\mid (j,i,s)\in E\}$; note that $G^I=G^{V\setminus I}$ and $(G^I)^I=G$. We say that $G$ is \EM{switch equivalent} to $H$ if $H=G^I$ for some $I\subseteq V$. Obviously, $G$ and $G^I$ have the same underlying digraph. Note also that $C$ is a cycle in $G$ if and only if $C^I$ is a cycle in $G^I$, and $C$ and $C^I$ have the same sign. Thus if $G$ has no positive (negative) cycles then every switch of $G$ has no positive (negative) cycles: this property is invariant by switch. The \EM{symmetric version} of $G$ is the signed digraph $\EMM{G^s}=(V,E^s)$ where $\EMM{E^s}=E\cup \{(i,j,s)\mid (j,i,s)\in E\}$. A well-known result concerning switch is the following adaptation of Harary's theorem \cite{H53}.

\begin{proposition}\label{pro:harary}
A signed digraph $G$ is switch equivalent to a full-positive signed digraph if and only if $G^s$ has no negative cycle. Furthermore, if $G$ is strong, then $G^s$ has no negative cycle if and only if $G$ has no negative cycle.
\end{proposition}

There is an analogue of the switch operation for BNs. Let $f\in F(V)$ and $I\subseteq V$. The \EM{$I$-switch} of $f$ is the BN $h\in F(V)$ defined by $h(x)=f(x+e_I)+e_I$ for all configurations $x$ on $V$; note that if $h$ is the $I$-switch of $f$ then $f$ is the $I$-switch of $h$. The analogy comes from the first point of the following easy property.

\begin{proposition}\label{pro:BN_switch}
If $h$ is the $I$-switch of $f$, then
\begin{itemize}
\item $G(h)$ is the $I$-switch of $G(f)$,
\item $\Gamma(h)$ is isometric to $\Gamma(f)$, with the isometry $x\mapsto x+e_I$.
\end{itemize}
\end{proposition}

\section{Intersections between positive and negative cycles}\label{sec:intersection_positive_negative}

What can we say about non-separating and non-trap-separating signed digraphs? If $G$ is non-separating or non-trap-separating then it is both non-fixing and non-converging. Hence, by \cref{thm:bib}, $G$ has both a positive and a negative cycle. Can we say something more? In this section, we provide the following answers: if $G$ is non-separating then it has intersecting cycles with opposite signs, and if it is non-trap-separating, then it has a path from a negative cycle to a positive cycle. 

\begin{theorem}\label{thm:sep}
If $G$ has no two intersecting cycles with opposite signs, then $G$ is separating.
\end{theorem}

\begin{theorem}\label{thm:trap-sep}
If $G$ has no path from a negative cycle to a positive cycle, then $G$ is trap-separating. 
\end{theorem}

Note the condition of \cref{thm:trap-sep} is stronger than the condition of \cref{thm:sep}: if a vertex $i$ meets both a positive cycle $C^+$ and a negative cycle $C^-$, then the trivial graph with $i$ as single vertex is a path (of length zero) from $C^-$ to $C^+$ (and from $C^+$ to $C^-$). Examples below show that the condition of \cref{thm:sep} (resp. \cref{thm:trap-sep}) is not sufficient to guarantee that $G$ is trap-separating (resp. trapping).  

\begin{example}\label{ex:not-trap-sep}
Let $f\in F(4)$ be defined by $f_1(x)=\bar x_3$, $f_2(x)=\bar x_1$, $f_3(x)=\bar x_2$ and $f_4(x)=x_1x_2x_3\lor x_4x_1\lor x_4x_2\lor x_4x_3$. Then $\Gamma(f)$ is separating and not trap-separating, and $G(f)$ has exactly two cycles, of opposite signs, which are vertex disjoint hence it satisfies the condition of \cref{thm:sep}.
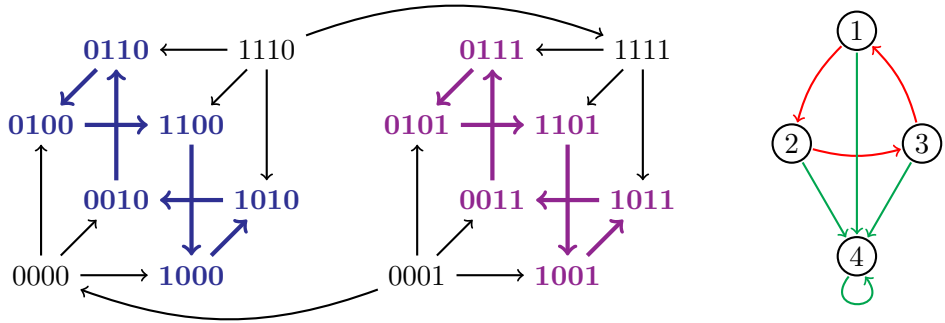
\[
\begin{array}{c}
\begin{tikzpicture}
\pgfmathparse{1}
\node (0000) at (0,0){$0000$};
\node (0010) at (1,1){{\boldmath \textcolor{Blue}{$0010$} \unboldmath}};
\node (0100) at (0,2){{\boldmath \textcolor{Blue}{$0100$} \unboldmath}};
\node (0110) at (1,3){{\boldmath \textcolor{Blue}{$0110$} \unboldmath}};
\node (1000) at (2,0){{\boldmath \textcolor{Blue}{$1000$} \unboldmath}};
\node (1010) at (3,1){{\boldmath \textcolor{Blue}{$1010$} \unboldmath}};
\node (1100) at (2,2){{\boldmath \textcolor{Blue}{$1100$} \unboldmath}};
\node (1110) at (3,3){$1110$};
\node (0001) at (5,0){$0001$};
\node (0011) at (6,1){{\boldmath \textcolor{Plum}{$0011$} \unboldmath}};
\node (0101) at (5,2){{\boldmath \textcolor{Plum}{$0101$} \unboldmath}};
\node (0111) at (6,3){{\boldmath \textcolor{Plum}{$0111$} \unboldmath}};
\node (1001) at (7,0){{\boldmath \textcolor{Plum}{$1001$} \unboldmath}};
\node (1011) at (8,1){{\boldmath \textcolor{Plum}{$1011$} \unboldmath}};
\node (1101) at (7,2){{\boldmath \textcolor{Plum}{$1101$} \unboldmath}};
\node (1111) at (8,3){$1111$};
\path[thick,->,draw,black]
(0100) edge[ultra thick,Blue] (1100)
(1100) edge[ultra thick,Blue] (1000)
(1000) edge[ultra thick,Blue] (1010)
(1010) edge[ultra thick,Blue] (0010)
(0010) edge[ultra thick,Blue] (0110)
(0110) edge[ultra thick,Blue] (0100)
(1110) edge (0110)
(1110) edge (1010)
(1110) edge (1100)
(1110) edge[bend left=20] (1111)
(0000) edge (1000)
(0000) edge (0100)
(0000) edge (0010)
(0101) edge[ultra thick,Plum] (1101)
(1101) edge[ultra thick,Plum] (1001)
(1001) edge[ultra thick,Plum] (1011)
(1011) edge[ultra thick,Plum] (0011)
(0011) edge[ultra thick,Plum] (0111)
(0111) edge[ultra thick,Plum] (0101)
(1111) edge (0111)
(1111) edge (1011)
(1111) edge (1101)
(0001) edge (1001)
(0001) edge (0101)
(0001) edge (0011)
(0001) edge[bend left=20] (0000)
;
\end{tikzpicture}
\end{array}
\qquad
\begin{array}{c}
\begin{tikzpicture}
\node[outer sep=1,inner sep=2,circle,draw,thick] (1) at (90:1){$1$};
\node[outer sep=1,inner sep=2,circle,draw,thick] (2) at ({90+120}:1){$2$};
\node[outer sep=1,inner sep=2,circle,draw,thick] (3) at ({90-120}:1){$3$};
\node[outer sep=1,inner sep=2,circle,draw,thick] (4) at (-90:2){$4$};
\draw[Green,->,thick] (4.-112) .. controls (-100:2.8) and (-80:2.8) .. (4.-68);
\path[->,thick]
(1) edge[red,bend right=15] (2)
(2) edge[red,bend right=15] (3)
(3) edge[red,bend right=15] (1)
(1) edge[Green] (4)
(2) edge[Green] (4)
(3) edge[Green] (4)
;
\end{tikzpicture}
\end{array}
\]
\end{example}

\begin{example}\label{ex:non-trapping}
Let $f\in F(4)$ be defined by $f_1(x)=\bar x_3$, $f_2(x)=\bar x_1$, $f_3(x)=\bar x_2$ and $f_4(x)=x_1x_2x_3$. Since $G(f)$ has no positive cycle, $\Gamma(f)$ has a unique attractor~$A$, but $[A]$ is not a trap space: $x_4=0$ for all $x\in A$ but $\Gamma(f)$ has an arc from $1010$ to $1011$. Hence $\Gamma(f)$ is converging, and thus trap-separating, but not trapping. Furthermore, since $G(f)$ has no positive cycle, it satisfies the condition of \cref{thm:trap-sep}.
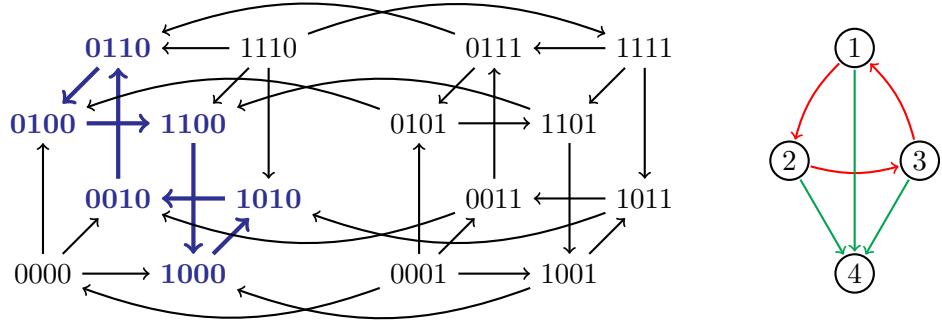
\[
\begin{array}{c}
\begin{tikzpicture}
\pgfmathparse{1}
\node (0000) at (0,0){$0000$};
\node (0010) at (1,1){{\boldmath \textcolor{Blue}{$0010$} \unboldmath}};
\node (0100) at (0,2){{\boldmath \textcolor{Blue}{$0100$} \unboldmath}};
\node (0110) at (1,3){{\boldmath \textcolor{Blue}{$0110$} \unboldmath}};
\node (1000) at (2,0){{\boldmath \textcolor{Blue}{$1000$} \unboldmath}};
\node (1010) at (3,1){{\boldmath \textcolor{Blue}{$1010$} \unboldmath}};
\node (1100) at (2,2){{\boldmath \textcolor{Blue}{$1100$} \unboldmath}};
\node (1110) at (3,3){$1110$};
\node (0001) at (5,0){$0001$};
\node (0011) at (6,1){$0011$};
\node (0101) at (5,2){$0101$};
\node (0111) at (6,3){$0111$};
\node (1001) at (7,0){$1001$};
\node (1011) at (8,1){$1011$};
\node (1101) at (7,2){$1101$};
\node (1111) at (8,3){$1111$};
\path[thick,->,draw,black]
(0100) edge[ultra thick,Blue] (1100)
(1100) edge[ultra thick,Blue] (1000)
(1000) edge[ultra thick,Blue] (1010)
(1010) edge[ultra thick,Blue] (0010)
(0010) edge[ultra thick,Blue] (0110)
(0110) edge[ultra thick,Blue] (0100)
(1110) edge (0110)
(1110) edge (1010)
(1110) edge (1100)
(0000) edge (1000)
(0000) edge (0100)
(0000) edge (0010)
(0101) edge (1101)
(1101) edge (1001)
(1001) edge (1011)
(1011) edge (0011)
(0011) edge (0111)
(0111) edge (0101)
(1111) edge (0111)
(1111) edge (1011)
(1111) edge (1101)
(0001) edge (1001)
(0001) edge (0101)
(0001) edge (0011)
(0001) edge[bend left=20] (0000)
(0011) edge[bend left=20] (0010)
(0101) edge[bend right=20] (0100)
(0111) edge[bend right=20] (0110)
(1001) edge[bend left=20] (1000)
(1011) edge[bend left=20] (1010)
(1101) edge[bend right=20] (1100)
(1110) edge[bend left=20] (1111)

;
\end{tikzpicture}
\end{array}
\qquad
\begin{array}{c}
\begin{tikzpicture}
\node[outer sep=1,inner sep=2,circle,draw,thick] (1) at (90:1){$1$};
\node[outer sep=1,inner sep=2,circle,draw,thick] (2) at ({90+120}:1){$2$};
\node[outer sep=1,inner sep=2,circle,draw,thick] (3) at ({90-120}:1){$3$};
\node[outer sep=1,inner sep=2,circle,draw,thick] (4) at (-90:2){$4$};
\path[->,thick]
(1) edge[red,bend right=15] (2)
(2) edge[red,bend right=15] (3)
(3) edge[red,bend right=15] (1)
(1) edge[Green] (4)
(2) edge[Green] (4)
(3) edge[Green] (4)
;
\end{tikzpicture}
\end{array}
\]
\end{example}

\medskip
The strategy to prove \cref{thm:sep} is roughly the following. Suppose that any two intersecting cycles of $G$ have the same sign. Then, in each strong component $H$, all the cycles have the same sign and thus, by \cref{thm:bib}, $H$ is fixing or converging, and thus separating. This suggests a proof by induction on the number of strong components, the base case ($G$ itself is strong) being given by the above argument. However, if all the strong components of $G$ are separating, then $G$ is not necessarily separating, as showed by the following examples (note that \cref{ex:not-trap-sep} and \cref{ex:non-trapping} show that, if all the strong components of $G$ are trap-separating or trapping, then $G$ is not necessarily trap-separating or trapping). 

\begin{example}
\label{ex:not-sep}
Let $f\in F(3)$ be defined by $f_1(x)=\bar x_1$, $f_2(x)=\bar x_1 x_3 \lor x_2 \bar x_3$ and $f_3(x)=x_1 x_2 \lor \bar x_2 x_3$. Then $\Gamma(f)$ is non-separating and $G(f)$ has exactly two strong components, $G[\{1\}]$ and $G[\{2,3\}]$, both converging (the second trivially so, since $F(G[\{2,3\}])=\emptyset$).
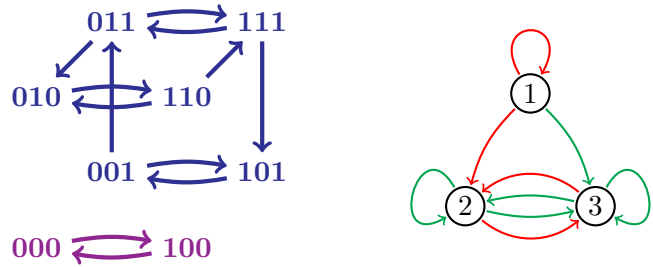
\[
\begin{array}{c}
\begin{tikzpicture}
\pgfmathparse{1}
\node (000) at (0,0){{\boldmath \textcolor{Plum}{$000$} \unboldmath}};
\node (001) at (1,1){{\boldmath \textcolor{Blue}{$001$} \unboldmath}};
\node (010) at (0,2){{\boldmath \textcolor{Blue}{$010$} \unboldmath}};
\node (011) at (1,3){{\boldmath \textcolor{Blue}{$011$} \unboldmath}};
\node (100) at (2,0){{\boldmath \textcolor{Plum}{$100$} \unboldmath}};
\node (101) at (3,1){{\boldmath \textcolor{Blue}{$101$} \unboldmath}};
\node (110) at (2,2){{\boldmath \textcolor{Blue}{$110$} \unboldmath}};
\node (111) at (3,3){{\boldmath \textcolor{Blue}{$111$} \unboldmath}};
\path[ultra thick,->,draw,Blue]
(000) edge[bend left=10,Plum] (100)
(001) edge (011)
(001) edge[bend left=10] (101)
(010) edge[bend left=10] (110)
(011) edge[bend left=10] (111)
(011) edge (010)
(100) edge[bend left=10,Plum] (000)
(101) edge[bend left=10] (001)
(110) edge (111)
(110) edge[bend left=10] (010)
(111) edge (101)
(111) edge[bend left=10] (011)
;
\end{tikzpicture}
\end{array}
\qquad
\begin{array}{c}
\begin{tikzpicture}
\node[outer sep=1,inner sep=2,circle,draw,thick] (1) at (90:1){$1$};
\node[outer sep=1,inner sep=2,circle,draw,thick] (2) at ({90+120}:1){$2$};
\node[outer sep=1,inner sep=2,circle,draw,thick] (3) at ({90-120}:1){$3$};
\draw[red,->,thick] (1) to [out=120,in=60,looseness=8] (1);
\draw[Green,->,thick] (2) to [out=120,in=210,looseness=6] (2);
\draw[Green,->,thick] (3) to [out=60,in=-30,looseness=6] (3);
\path[->,thick]
(1) edge[red,bend right=15,->] (2)
(1) edge[Green,bend left=15] (3)
(2) edge[Green,bend right=13] (3)
(2) edge[red,bend right=40,->] (3)
(3) edge[Green,bend right=13] (2)
(3) edge[red,bend right=40,->] (2)
;
\end{tikzpicture}
\end{array}
\]
\end{example}

\begin{example}
\label{ex:not-sep-2}
We give a second example, where the set $F(H)$ is non-empty for all strong components $H$ of $G$.
Consider $f\in F(4)$ defined by $f_1(x)=x_1x_3 \vee \bar x_2 x_3$, $f_2(x)= x_2 x_3 \vee \bar x_1 x_3$, $f_3(x)= x_1x_2 \vee \bar x_4$ and $f_4(x)=\bar x_4$. Then $\Gamma(f)$ is non-separating as shown in the figure below. $G(f)$ has two strong components, $G[\{4\}]$, which is converging, and $G[\{1,2,3\}]$ which is separating by~\cref{thm:fvs2}, since it has feedback number two and all of its negative cycles contain all three vertices.
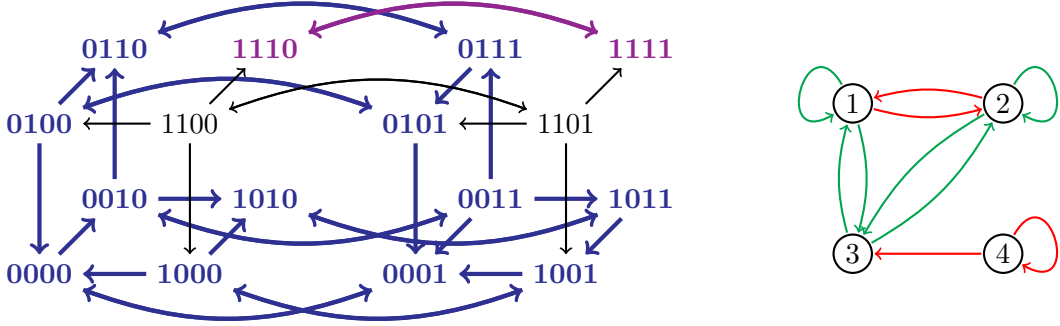
\[
\begin{array}{c}
\begin{tikzpicture}
\pgfmathparse{1}
\node (0000) at (0,0){{\boldmath \textcolor{Blue}{$0000$} \unboldmath}};
\node (0010) at (1,1){{\boldmath \textcolor{Blue}{$0010$} \unboldmath}};
\node (0100) at (0,2){{\boldmath \textcolor{Blue}{$0100$} \unboldmath}};
\node (0110) at (1,3){{\boldmath \textcolor{Blue}{$0110$} \unboldmath}};
\node (1000) at (2,0){{\boldmath \textcolor{Blue}{$1000$} \unboldmath}};
\node (1010) at (3,1){{\boldmath \textcolor{Blue}{$1010$} \unboldmath}};
\node (1100) at (2,2){$1100$};
\node (1110) at (3,3){{\boldmath \textcolor{Plum}{$1110$} \unboldmath}};
\node (0001) at (5,0){{\boldmath \textcolor{Blue}{$0001$} \unboldmath}};
\node (0011) at (6,1){{\boldmath \textcolor{Blue}{$0011$} \unboldmath}};
\node (0101) at (5,2){{\boldmath \textcolor{Blue}{$0101$} \unboldmath}};
\node (0111) at (6,3){{\boldmath \textcolor{Blue}{$0111$} \unboldmath}};
\node (1001) at (7,0){{\boldmath \textcolor{Blue}{$1001$} \unboldmath}};
\node (1011) at (8,1){{\boldmath \textcolor{Blue}{$1011$} \unboldmath}};
\node (1101) at (7,2){$1101$};
\node (1111) at (8,3){{\boldmath \textcolor{Plum}{$1111$} \unboldmath}};
\path[thick,->,draw,black]
(0000) edge[ultra thick,bend right=20,Blue] (0001)
(0000) edge[ultra thick,Blue] (0010)
(0001) edge[ultra thick,bend left=20,Blue] (0000)
(0010) edge[ultra thick,bend right=20,Blue] (0011)
(0010) edge[ultra thick,Blue] (0110)
(0010) edge[ultra thick,Blue] (1010)
(0011) edge[ultra thick,Blue] (0001)
(0011) edge[ultra thick,bend left=20,Blue] (0010)
(0011) edge[ultra thick,Blue] (0111)
(0011) edge[ultra thick,Blue] (1011)
(0100) edge[ultra thick,bend left=20,Blue] (0101)
(0100) edge[ultra thick,Blue] (0110)
(0100) edge[ultra thick,Blue] (0000)
(0101) edge[ultra thick,Blue] (0001)
(0101) edge[ultra thick,bend right=20,Blue] (0100)
(0110) edge[ultra thick,bend left=20,Blue] (0111)
(0111) edge[ultra thick,Blue] (0101)
(0111) edge[ultra thick,bend right=20,Blue] (0110)
(1000) edge[ultra thick,Blue] (1010)
(1000) edge[ultra thick,bend right=20,Blue] (1001)
(1000) edge[ultra thick,Blue] (0000)
(1001) edge[ultra thick,Blue] (0001)
(1001) edge[ultra thick,bend left=20,Blue] (1000)
(1010) edge[ultra thick,bend right=20,Blue] (1011)
(1011) edge[ultra thick,Blue] (1001)
(1011) edge[ultra thick,bend left=20,Blue] (1010)
(1100) edge (0100)
(1100) edge (1000)
(1100) edge (1110)
(1100) edge[bend left=20] (1101)
(1101) edge[bend right=20] (1100)
(1101) edge (0101)
(1101) edge (1001)
(1101) edge (1111)
(1110) edge[ultra thick,bend left=20,Plum] (1111)
(1111) edge[ultra thick,bend right=20,Plum] (1110)
;
\end{tikzpicture}
\end{array}
\qquad
\begin{array}{c}
\begin{tikzpicture}
\node[outer sep=1,inner sep=2,circle,draw,thick] (1) at (0,0){$1$};
\node[outer sep=1,inner sep=2,circle,draw,thick] (2) at (2,0){$2$};
\node[outer sep=1,inner sep=2,circle,draw,thick] (3) at (0,-2){$3$};
\node[outer sep=1,inner sep=2,circle,draw,thick] (4) at (2,-2){$4$};
\draw[red,->,thick] (4) to [out=60,in=-30,looseness=6] (4);
\draw[Green,->,thick] (1) to [out=120,in=210,looseness=6] (1);
\draw[Green,->,thick] (2) to [out=60,in=-30,looseness=6] (2);
\path[->,thick]
(1) edge[red,bend right=15] (2)
(1) edge[Green,bend left=15] (3)
(2) edge[red,bend right=15] (1)
(2) edge[Green,bend right=15] (3)
(3) edge[Green,bend left=15] (1)
(3) edge[Green,bend right=15] (2)
(4) edge[red] (3)
;
\end{tikzpicture}
\end{array}
\]
\end{example} 

On the other hand, if all the cycles of $H$ have the same sign, then $H$ is more than separating, it is ``robustly'' separating (a formal definition will be given below), and it turns out that if each strong component of $G$ is ``robustly'' separating, then $G$ is separating and we are done. 

\medskip
Formally, $G$ is \EM{robustly separating} if, for any non-empty set $F$ of BNs such that $G(f)$ is a spanning subgraph of $G$ for all $f\in F$, the (joint) union $\bigcup_{f\in F} \Gamma(f)$ is separating (a spanning subgraph of $G$ is a subgraph of $G$ with vertex set $V$). We define similarly the notions of \EM{robustly converging} and \EM{robustly trapping}.

\begin{example}
  The graph $G[\{1,2,3\}]$ of \cref{ex:not-sep-2} is separating but not robustly separating:
  the maps $f(x)=(\bar x_2 \lor  x_1 x_3, x_2 \bar x_1 \lor x_3 \bar x_1, x_1 \lor  x_2)$ and $g(x)=(x_1 x_3 \lor  x_3 \bar x_2, x_3 \lor x_2 \bar x_1, x_1 x_2)$ are fixing, but the union of $\Gamma(f)$ and $\Gamma(g)$ is not separating. The asynchronous graphs of $f$ and $g$ and their union are as follows:
 
  \[
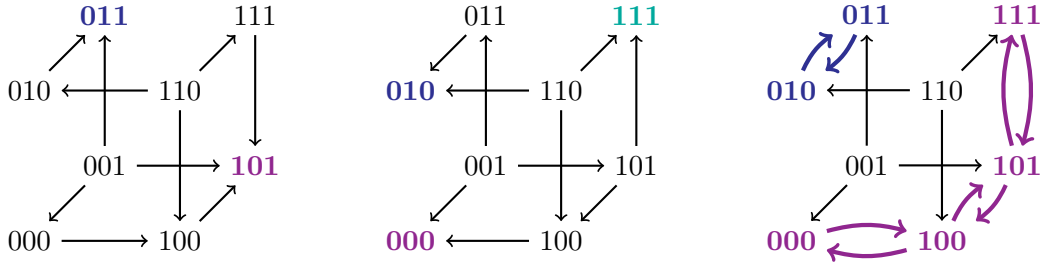
  \begin{array}{c}
  \begin{tikzpicture}
  \pgfmathparse{1}
  \node (000) at (0,0){{$000$}};
  \node (001) at (1,1){{$001$}};
  \node (010) at (0,2){$010$};
  \node (011) at (1,3){{\boldmath \textcolor{Blue}{$011$} \unboldmath}};
  \node (100) at (2,0){$100$};
  \node (101) at (3,1){{\boldmath \textcolor{Plum}{$101$} \unboldmath}};
  \node (110) at (2,2){$110$};
  \node (111) at (3,3){$111$};
  \path[thick,->,draw]
  (000) edge (100)
  (001) edge (101)
  (001) edge (000)
  (001) edge (011)
  (010) edge (011)
  (100) edge (101)
  (110) edge (100)
  (110) edge (111)
  (110) edge (010)
  (111) edge (101)
  ;
  \end{tikzpicture}
  \end{array}
  \qquad
  \begin{array}{c}
  \begin{tikzpicture}
  \pgfmathparse{1}
  \node (000) at (0,0){{\boldmath \textcolor{Plum}{{$000$}} \unboldmath}};
  \node (001) at (1,1){{$001$}};
  \node (010) at (0,2){{\boldmath \textcolor{Blue}{$010$} \unboldmath}};
  \node (011) at (1,3){$011$};
  \node (100) at (2,0){$100$};
  \node (101) at (3,1){$101$};
  \node (110) at (2,2){$110$};
  \node (111) at (3,3){{\boldmath \textcolor{Emerald}{$111$} \unboldmath}};
  \path[thick,->,draw]
  (001) edge (101)
  (001) edge (000)
  (001) edge (011)
  (011) edge (010)
  (100) edge (000)
  (101) edge (100)
  (101) edge (111)
  (110) edge (100)
  (110) edge (111)
  (110) edge (010)
  ;
  \end{tikzpicture}
  \end{array}
  \qquad
  \begin{array}{c}
  \begin{tikzpicture}
  \pgfmathparse{1}
  \node (000) at (0,0){{\boldmath \textcolor{Plum}{$000$} \unboldmath}};
  \node (001) at (1,1){$001$};
  \node (010) at (0,2){{\boldmath \textcolor{Blue}{$010$} \unboldmath}};
  \node (011) at (1,3){{\boldmath \textcolor{Blue}{$011$} \unboldmath}};
  \node (100) at (2,0){{\boldmath \textcolor{Plum}{$100$} \unboldmath}};
  \node (101) at (3,1){{\boldmath \textcolor{Plum}{$101$} \unboldmath}};
  \node (110) at (2,2){$110$};
  \node (111) at (3,3){{\boldmath \textcolor{Plum}{$111$} \unboldmath}};
  \path[thick,->,draw]
  (000) edge[bend left=15,Plum,ultra thick] (100)
  (001) edge (101)
  (001) edge (000)
  (001) edge (011)
  (010) edge[bend left=15,Blue,ultra thick] (011)
  (011) edge[bend left=15,Blue,ultra thick] (010)
  (100) edge[bend left=15,Plum,ultra thick] (000)
  (100) edge[bend left=15,Plum,ultra thick] (101)
  (101) edge[bend left=15,Plum,ultra thick] (100)
  (101) edge[bend left=15,Plum,ultra thick] (111)
  (110) edge (100)
  (110) edge (111)
  (110) edge (010)
  (111) edge[bend left=15,Plum,ultra thick] (101)
  ;
  \end{tikzpicture}
  \end{array}
  \]
\end{example}

\medskip
Below, we prove that if $G$ is strong and has only negative (positive) cycles, then $G$ is robustly converging (trapping) and thus robustly separating. 

\begin{lemma}\label{lem:robustly_converging}
If all the cycles of $G$ are negative, then $G$ is robustly converging.
\end{lemma}

\begin{proof}
Suppose that all the cycles of $G$ are negative. Let $F$ be a set of BNs such that $G(f)$ is a spanning subgraph of $G$ for all $f\in F$, and let $\Gamma=\bigcup_{f\in F} \Gamma(f)$. Suppose that $\Gamma$ has two distinct attractors $A$ and $B$, and let $f\in F$. Since $\Gamma(f)$ is a subgraph of $\Gamma$, $A$ and $B$ are trap sets of $\Gamma(f)$, and thus $\Gamma(f)$ has at least two distinct attractors, one included in $A$, and the other included in $B$. But since $G(f)$ is a subgraph of $G$, it has only negative cycles. Thus $\Gamma(f)$ is converging by \cref{thm:bib}, and we obtain a contradiction. This proves that $\Gamma$ is converging. 
\end{proof}

To treat the case where all the cycles are positive, we need the following lemma. 

\begin{lemma}\label{lem:monotone_union}
Let $f^1,\dots,f^\ell$ be $\ell$ monotone BNs with component set $V$. Let
\[
\Gamma=\Gamma(f^1)\cup\cdots\cup \Gamma(f^\ell)
\]
be the joint union of the corresponding asynchronous graphs. Then $\Gamma$ is trapping.
\end{lemma}

\begin{proof}
We need:
\begin{quote}
(1) {\em If $\Gamma$ has no decreasing arc starting from $a$ then $[a,\ONE]$ is a trap space, and if $\Gamma$ has no increasing arc starting from $b$ then $[\ZERO,b]$ is a trap space.}

\smallskip
Suppose that $[a,\ONE]$ is not a trap space. Then $\Gamma$ has an arc $x\to y$ leaving $[a,\ONE]$. Let $i$ be the direction of this arc. Then $x_i=a_i=1$ and $y_i=0$. Thus $f^k_i(x)=0$ for some $1\leq k\leq \ell$, and since $f^k$ is monotone and $a\leq x$, we have $f^k_i(a)\leq f^k_i(x)=0$. Since $a_i=1$, we deduce that $\Gamma(f^k)$, and thus $\Gamma$, has an arc $a\to a'$ with $a'_i=0$, which is decreasing. This proves the first assertion, and the second is similar.
\end{quote}

Let $A$ be an attractor of $\Gamma$.

\begin{quote}
(2) {\em $A$ has a unique minimal element and a unique maximal element.}

\smallskip
Consider $a,a'$ minimal elements of $A$.
Then $\Gamma$ has no decreasing arc starting from $a$ thus, by (1), $[a,\ONE]$ is a trap space, as is $[a',\ONE]$ with the same proof.
Since $a' \in A$, we have $a'\in [a,\ONE]$, and therefore $a'\geq a$, and symmetrically $a \geq a'$, which proves $a=a'$. We prove similarly that $A$ has a unique maximal element.
\end{quote}

Let $a$ and $b$ be the minimal and maximal element of $A$. Then $[A]=[a,b]$ and, by (1), $[a,\ONE]$ and $[\ZERO,b]$ are trap spaces.

\begin{quote}
(3) {\em $[A]$ is a trap space of $\Gamma$.}

\smallskip
Suppose that $x\to y$ is an arc leaving $[A]$, and let $i$ be the direction of this arc. Then $x_i=a_i=b_i\neq y_i$. If $x_i=1$ we deduce that $x\to y$ leaves $[a,\ONE]$ and if $x_i=0$ we deduce that $x\to y$ leaves $[\ZERO,b]$, and in both cases we obtain a contradiction.
\end{quote}

\begin{quote}
(4) {\em $\Gamma$ has a path from every configuration in $[A]$ to $A$.}

\smallskip
For every $x\in [A]$ we prove, by induction on $d(x,b)$, that $\Gamma$ has a path from $x$ to $b$. If $d(x,b)=0$ then there is nothing to prove. So suppose that $d(x,b)>0$. If $\Gamma$ has no increasing arc starting from $x$ then, by (1), $[\ZERO,x]$ is a trap space, which contains $a$ but not $b$. Since $a,b\in A$, $\Gamma$ has a path from $a$ to $b$, and thus $[\ZERO,x]$ is not a trap space, a contradiction. Hence $\Gamma$ has an increasing arc $x\to y$. By (3), $[A]$ is a trap space, so $y\leq b$, and since $x\leq y$ we have $d(y,b)<d(x,b)$. Thus, by induction, $\Gamma$ has a path from $y$ to $b$, and by adding $x\to y$ to this path, we obtain the desired path.
\end{quote}

Suppose that $\Gamma$ has an attractor $B\neq A$. For $X\in\{A,B\}$, let $R(X)$ be the set of configurations $x$ such that $\Gamma$ has a path from $x$ to $X$. By (3) and (4), we have $[A]=\langle A\rangle$ and $[A]\subseteq R(A)$ and, similarly, $[B]=\langle B\rangle$ and $[B]\subseteq R(B)$. Suppose, for a contradiction, that there is $x\in [A]\cap [B]$. Then $x\in  R(A)$ and since $x\in \langle B\rangle$ we have $R(A)\subseteq \langle B\rangle=[B]$. Since $A\subseteq R(A)$ we have $A\subseteq  [B]\subseteq R(B)$, and thus $\Gamma$ has a path from $A$ to $B$, which is a contradiction since $A,B$ are distinct attractors. Thus $[A]\cap [B]=\emptyset$, and thus $\Gamma$ is trapping.
\end{proof}

We deduce:

\begin{lemma}\label{lem:robustly_trapping}
If $G$ is strong and has only positive cycles, then $G$ is robustly trapping.
\end{lemma}

\begin{proof}
Suppose that $G$ is strong and has only positive cycles. By \cref{pro:harary} and \cref{pro:BN_switch} we can suppose that $G$ is full-positive. Let $F$ be a set of BNs such that $G(f)$ is a spanning subgraph of $G$ for all $f\in F$, and let $\Gamma=\bigcup_{f\in F} \Gamma(f)$. For every $f\in F$, since $G(f)$ is full-positive, $f$ is monotone and we deduce from Lemma~\ref{lem:monotone_union} that $\Gamma$ is trapping.
\end{proof}

Going back to the proof of \cref{thm:sep}, we now know that if any two intersecting cycles of $G$ have the same sign, then each strong component of $G$ is robustly separating. It remains to prove that this implies that $G$ is separating. For that, we need a decomposition technique for non-strong signed digraphs. If $G$ is not strong, then there is a partition $(I_1,I_2)$ of the vertices such that $G$ has no arc from $I_2$ to $I_1$. Given $f\in F(G)$, we then show that each attractor of $\Gamma(f)$ can be regarded as the Cartesian product of an asynchronous attractor of the ``restriction'' of $f$ on $I_1$ and a union of asynchronous attractors of BNs whose signed interaction digraph is a spanning subgraph of $G[I_2]$. (The union involved in the definition of robustly separating is actually motivated by the union involved in this decomposition.) The details follow.

\medskip
Let $f\in F(V)$, and $(I_1,I_2)$ a partition of $V$, without empty part. We identify $\B^V$ with $\B^{I_1}\times \B^{I_2}$. Thus we regard each configuration $x$ on $V$ has a pair $x=(x_{I_1},x_{I_2})$. We denote by $f^1$  the subnetwork of $f$ induced by $[(\ZERO,\ZERO),(\ONE,\ZERO)]$ and set $\Gamma^1=\Gamma(f^1)$. Hence $f^1$ is obtained by fixing to $0$ each component in $I_2$.  Next, for all configurations $x$ on $I_1$, we denote by $f^x$ be subnetwork of $f$ induced by $[(x,\ZERO),(x,\ONE)]$. Hence $f^x$ is obtained by fixing to $x_i$ each component $i$ in $I_1$. Let $A$ be an attractor of $\Gamma(f)$. We set:
\[
A^1=\{a_{I_1}\mid a\in A\},\qquad A^2=\{a_{I_2}\mid a\in A\},\qquad \Gamma^2_A=\bigcup_{x\in A^1} \Gamma(f^x).
\]

\begin{lemma}\label{lem:decomposition}
Let $f\in F(V)$. Let $(I_1,I_2)$ be a partition of $V$ without empty part. Suppose that $G(f)$ has no arc from $I_2$ to $I_1$. For every attractor $A$ of $\Gamma(f)$:
\begin{itemize}
\item $A=A^1\times A^2$,
\item $A^1$ is an attractor of $\Gamma^1$,
\item $A^2$ is an attractor of $\Gamma^2_A$.
\end{itemize}
\end{lemma}

\begin{proof}
Let $x$ be a configuration on $I_1$ and let $a$ be a configuration on $I_2$. Since $G(f)$ has no arc from $I_2$ to $I_1$, we have $f(x,a)_{I_1}=f(x,\ZERO)_{I_1}=f^1(x)$ and we deduce that 
\begin{quote}
(1) {\em $x\to y$ is an arc of $\Gamma^1$ if and only if $(x,a)\to (y,a)$ is an arc of $\Gamma(f)$.}
\end{quote}

\medskip
It follows that $A^1$ is an attractor of $\Gamma^1$. Indeed, let $x\in A^1$ and let $a$ be a configuration on $I_2$ such that $(x,a)\in A$. If $x\to y$ is an arc of $\Gamma^1$ then, by (1), $(x,a)\to (y,a)$ is an arc of $\Gamma(f)$ and since $(x,a)\in A$ we have $(y,a)\in A$ and thus $y\in A^1$. So $A^1$ is a trap set of $\Gamma^1$. Let $B^1$ be a strict subset of $A^1$. Let $x\in B^1$, $y\in A^1\setminus B^1$, and let $a,b$ be configurations on $I_2$ such that $(x,a),(y,b)\in A$. Since $A$ is an attractor of $\Gamma(f)$, there is a path from $(x,a)$ to $(y,b)$, and this path contains an arc $(z,c)\to (z',c)$ with $z\in B^1$ and $z'\in A^1\setminus B^1$. Then, by (1), $z\to z'$ is an arc of $\Gamma^1$ leaving $B^1$. Hence $A^1$ is an inclusion-minimal trap set of $\Gamma^1$, as desired.

\medskip
We now prove that $A=A^1\times A^2$. It is sufficient to prove that $A^1\times A^2\subseteq A$ since the other direction is clear. Let $(y,a)\in A^1\times A^2$. Since $a\in A^2$, there is $x\in A^1$ such that $(x,a)\in A$. Since $A^1$ is an attractor of $\Gamma^1$, $\Gamma^1$ has a path from $x$ to $y$ and we deduce from (1) that $\Gamma(f)$ has a path from $(x,a)$ to $(y,a)$, and thus $(y,a)\in A$.

\medskip
We finally prove that $A^2$ is an attractor of $\Gamma^2_A$. Suppose that $\Gamma^2_A$ has an arc $a\to b$ with $a\in A^2$. There is $x\in A^1$ such that $a\to b$ is an arc of $\Gamma(f^x)$, and we deduce that $(x,a)\to (x,b)$ is an arc of $\Gamma(f)$. Since $A=A^1\times A^2$, we have $(x,a)\in A$ and thus $(x,b)\in A$ and we deduce that $b\in A^2$. So $A^2$ is a trap set of $\Gamma^2_A$. Let $B^2$ be a strict subset of $A^2$. Let $a\in B^2$, $b\in A^2\setminus B^2$, and let $x,y$ be configurations on $I_1$ such that $(x,a),(y,b)\in A$. Since $A$ is an attractor of $\Gamma(f)$, there is a path from $(x,a)$ to $(y,b)$, and this path contains an arc $(z,c)\to (z,c')$ with $c\in B^2$ and $c'\in A^2\setminus B^2$. Since $z\in A^1$, $c\to c'$ is an arc of $\Gamma(f^z)$, and thus an arc of $\Gamma^2_A$ leaving $B^2$. Hence $A^2$ is an inclusion-minimal trap set of $\Gamma^2_A$, as desired.
\end{proof}

We deduce the following which, together with \cref{lem:robustly_converging} and \cref{lem:robustly_trapping}, implies \cref{thm:sep}.

\begin{lemma}\label{lem:decomposition_2}
If each strong component of $G$ is robustly separating, then $G$ is separating.
\end{lemma}

\begin{proof}
Suppose that every strong component of $G$ is robustly separating. We proceed by induction on the number of strong components. If $G$ is strong then the result is obvious. Otherwise, there is a partition $(I_1,I_2)$ of the vertices of $G$ such that $G$ has no arc from $I_2$ to $I_1$ and such that $G[I_2]$ is a strong component of $G$. Let $f\in F(G)$. Since $G(f^1)=G[I_1]$, by induction, $\Gamma^1$ is separating. Furthermore, for every attractor $A$ of $\Gamma(f)$ and $x\in A$, $G(f^x)$ is a spanning subgraph of $G[I_2]$, which is robustly separating, and we deduce that $\Gamma^2_A$ is separating. Let $A,B$ be distinct attractors of $\Gamma(f)$. By \cref{lem:decomposition}, we have $A=A^1\times A^2$ and $B=B^1\times B^2$. If $A^1\neq B^1$ then, by the same lemma, $A^1,B^1$ are distinct attractors of $\Gamma^1$, which is separating, thus $[A^1]\cap [B^1]=\emptyset$ and we deduce that $[A]\cap [B]=\emptyset$. Suppose now that $A^1=B^1$. Then, by the same lemma,  $A^2,B^2$ are distinct attractors of $\Gamma^2_A=\Gamma^2_B$, which is separating. Thus $[A^2]\cap [B^2]=\emptyset$ and we deduce that $[A]\cap [B]=\emptyset$. This proves that $\Gamma(f)$ is separating.
\end{proof}

The proof of \cref{thm:trap-sep} follows the same line and is easier. Let us say that $G$ is \EM{perfectly fixing} if each subgraph of $G$ is fixing. Clearly, if all the cycles of $G$ are positive, then $G$ is perfectly fixing. Suppose now that the conditions of \cref{thm:trap-sep} are satisfied, that is, $G$ has no path from a negative cycle to a positive cycle. Then each strong component is either perfectly fixing (if all the cycles are positive) or robustly converging (if all the cycles are negative, \cref{lem:robustly_converging}), and there is no arc from a robustly converging component to a perfectly fixing component. We prove below that this is enough for $G$ to be trap-separating, and this proves \cref{thm:trap-sep}. 

\begin{lemma}\label{lem:decomposition_3}
Suppose that each strong component of $G$ is either perfectly fixing or robustly converging, and that there is no arc from a robustly converging component to a robustly fixing component. Then $G$ is trap-separating. 
\end{lemma}

We need the following:

\begin{lemma}\label{lem:decomposition_4}
If each strong component of $G$ is perfectly fixing, then $G$ is perfectly fixing.
\end{lemma}

\begin{proof}
Suppose that every strong component of $G$ is perfectly fixing. We proceed by induction on the number of strong components. If $G$ is strong then the result is obvious. Otherwise, there is a partition $(I_1,I_2)$ of the vertices of $G$ such that $G$ has no arc from $I_2$ to $I_1$ and such that $G[I_2]$ is a strong component of $G$. Let $f\in F(G)$. Since $G(f^1)=G[I_1]$, by induction, $\Gamma^1$ is fixing. Let $A$ be an attractor of $\Gamma(f)$. By \cref{lem:decomposition}, we have $A=A^1\times A^2$ and $A^1$ is an attractor of $\Gamma^1$. Thus $A^1=\{a\}$ for some fixed point $a$ of $f^1$. By the same lemma, $A^2$ is an attractor of $\Gamma^2_{A^1}=\Gamma(f^{a})$. Since $G(f^{a})$ is a subgraph of $G[I_2]$, it is fixing, and thus $|A^2|=1$. We deduce that $|A|=1$, and thus $\Gamma(f)$ is fixing. Consequently, $G$ is fixing. Let $G'$ be a subgraph of $G$. Then each strong component of $G'$ is perfectly fixing and the argument above shows that $G'$ is fixing. Consequently, $G$ is perfectly fixing.
\end{proof}

\begin{lemma}\label{lem:decomposition_5}
Suppose that there is $I_2\subseteq V$ such that $G[I_2]$ is a robustly converging terminal strong component of $G$ and that $G\setminus I_2$ is trap-separating. Then $G$ is trap-separating. 
\end{lemma}

\begin{proof}
Let $f\in F(G)$ and $I_1=V\setminus I_2$. Since $G(f^1)=G[I_1]$, $\Gamma^1$ is trap-separating. Furthermore, for every attractor $A$ of $\Gamma(f)$ and $x\in A$, $G(f^x)$ is a spanning subgraph of $G[I_2]$, which is robustly converging, and we deduce that $\Gamma^2_A$ is converging. Let $A,B$ be distinct attractors of $\Gamma(f)$. By \cref{lem:decomposition}, we have $A=A^1\times A^2$ and $B=B^1\times B^2$. If $A^1\neq B^1$ then, by the same lemma, $A^1,B^1$ are distinct attractors of $\Gamma^1$, which is trap-separating, thus $\langle A^1\rangle\cap \langle B^1\rangle=\emptyset$. Consequently, $\langle A^1 \rangle \times\B^{I_2}$ and $\langle B^1\rangle\times  \B^{I_2}$ are disjoint trap spaces of $\Gamma(f)$ containing $A$ and $B$, and thus $\langle A\rangle\cap \langle B\rangle=\emptyset$. Suppose now that $A^1=B^1$. Then, by the same lemma, $A^2,B^2$ are distinct attractors of $\Gamma^2_A=\Gamma^2_B$, which is converging, a contradiction. This proves that $\Gamma(f)$ is trap-separating. 
\end{proof}

\begin{proof}[\BF{Proof of \cref{lem:decomposition_3}}]
We proceed by induction on the number of strong components. If $G$ is strong then $G$ is either (perfectly) fixing or (robustly) converging and thus $G$ is trap-separating. So suppose that $G$ is not strong. If all the strong components of $G$ are perfectly fixing, then, by \cref{lem:decomposition_4}, $G$ is fixing and thus trap-separating. So suppose that $G$ has a strong component which is robustly converging. Since there is no path from a robustly converging strong component to a perfectly fixing strong component, there is a partition $(I_1,I_2)$ of the vertices of $G$ such that $G$ has no arc from $I_2$ to $I_1$ and such that $G[I_2]$ is a robustly converging strong component of $G$. By induction hypothesis, $G[I_1]$ is trap-separating, and thus, by \cref{lem:decomposition_5}, $G$ is trap-separating.
\end{proof}

\section{Number of positive cycles}\label{sec:number-positive-cycles}

We have proved that if $G$ is non-separating, then it has a positive cycle intersecting a negative cycle. In this section, we prove the following, which says more concerning positive cycles. 

\begin{theorem}\label{thm:PFN}
If the positive feedback number of $G$ is at most one, then $G$ is separating.
\end{theorem}

\cref{ex:not-trap-sep} shows that, in the theorem, separating cannot be replaced by trap-separating. For the proof we need the following lemma.

\begin{lemma}\label{lem:Delta_A}
Let $f\in F(G)$ and $A,B$ be distinct attractors of $\Gamma(f)$. For every $i\in\Delta(A)$, $G\setminus i$ has a positive cycle. 
\end{lemma}

\begin{proof}
Let $i\in \Delta(A)$ and $b\in B$. Then there is $a\in A$ with $a_i=b_i$. Let $f'\in F(V)$ defined as follows: for all configurations $x$ on $V$, $f'_j(x)=f_j(x)$ for $j\neq i$ and $f'_i(x)=x_i$. Then $\Gamma(f')$ is the spanning subgraph of $\Gamma(f)$ obtained by deleting all the arcs in the direction $i$. Let $A'$ and $B'$ be attractors of $\Gamma(f')$ which are reachable in $\Gamma(f')$ from $a$ and $b$, respectively. Then $A'\subseteq A$ and $B'\subseteq B$ thus $A'\cap B'=\emptyset$. Furthermore, since $a_i=b_i$ we have $x_i=a_i=b_i$ for all $x\in A'\cup B'$. Let $(\alpha,\beta)\in A'\times B'$ with $\Delta(\alpha,\beta)$ minimum. Then $f'_j(\alpha)=\alpha_j$ and $f'_j(\beta)=\beta_j$ for all $j\in\Delta(\alpha,\beta)$. By \cref{lem:A08}, the subgraph of $G(f')$ induced by $\Delta(\alpha,\beta)$ has a positive cycle $C'$, which does not contain $i$ since $i\not\in \Delta(\alpha,\beta)$. Thus $C'$ is a positive cycle of $G(f')\setminus i=G\setminus i$. 
\end{proof}

\begin{proof}[\BF{Proof of \cref{thm:PFN}}]
If the positive feedback number of $G$ is zero, then $G$ is converging. So suppose that there is a vertex $i$ such that $G\setminus i$ has no positive cycle. Let $f\in F(G)$ and let $A,B$ be distinct attractors of $\Gamma(f)$. By \cref{lem:Delta_A}, we have $i\not\in\Delta(A)\cup\Delta(B)$.   Let $(a,b)\in A\times B$ with $\Delta(a,b)$ minimum. Then $f_j(a)=a_j$ and $f_j(b)=b_j$ for all $j\in\Delta(a,b)$. Hence, by \cref{lem:A08}, $G[\Delta(a,b)]$ has a positive cycle $C$. Since $C$ contains $i$ we have $a_i\neq b_i$. Since $i\not\in\Delta(A)\cup\Delta(B)$, we have $[A]\cap [B]=\emptyset$ and thus $\Gamma(f)$ is separating. 
\end{proof}

Hence, if $G$ is non-separating, then $G$ has at least two disjoint positive cycles or at least three positive cycles. The example below show that if $G$ is non-separating, then $G$ does not necessarily have two disjoint positive cycles. 

\begin{example}
Consider $f\in F(3)$ defined by $f_1(x)=x_2 \bar x_3 \lor x_3 \bar x_1 \lor x_3 \bar x_2$, $f_2(x)=x_1 \bar x_3 \lor x_3 \bar x_1 \lor x_3 \bar x_2$ and $f_3(x)=x_1 \bar x_2 \lor x_2 \bar x_1 \lor x_2 \bar x_3$. $\Gamma(f)$ is non-separating, and $G(f)$ does not have two disjoint positive cycles.
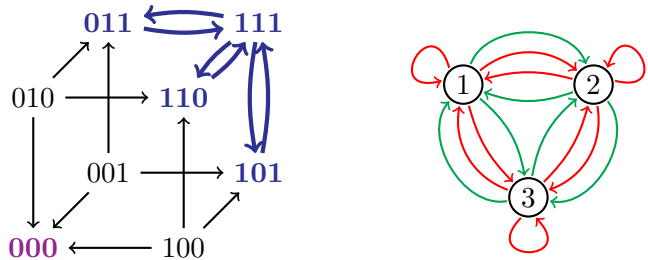
\[
\begin{array}{c}
\begin{tikzpicture}
\pgfmathparse{1}
\node (000) at (0,0){{\boldmath \textcolor{Plum}{$000$} \unboldmath}};
\node (001) at (1,1){$001$};
\node (010) at (0,2){$010$};
\node (011) at (1,3){{\boldmath \textcolor{Blue}{$011$} \unboldmath}};
\node (100) at (2,0){$100$};
\node (101) at (3,1){{\boldmath \textcolor{Blue}{$101$} \unboldmath}};
\node (110) at (2,2){{\boldmath \textcolor{Blue}{$110$} \unboldmath}};
\node (111) at (3,3){{\boldmath \textcolor{Blue}{$111$} \unboldmath}};
\path[thick,->,draw,black]
(001) edge (000)
(001) edge (011)
(001) edge (101)
(010) edge (000)
(010) edge (011)
(010) edge (110)
(011) edge[ultra thick,Blue,bend right=10] (111)
(101) edge[ultra thick,Blue,bend right=10] (111)
(110) edge[ultra thick,Blue,bend right=10] (111)
(100) edge (000)
(100) edge (110)
(100) edge (101)
(111) edge[ultra thick,Blue,bend right=10] (011)
(111) edge[ultra thick,Blue,bend right=10] (101)
(111) edge[ultra thick,Blue,bend right=10] (110)
;
\end{tikzpicture}
\end{array}
\qquad
\begin{array}{c}
\begin{tikzpicture}
\node[outer sep=1,inner sep=2,circle,draw,thick] (1) at ({-120-90}:1){$1$};
\node[outer sep=1,inner sep=2,circle,draw,thick] (2) at ({120-90}:1){$2$};
\node[outer sep=1,inner sep=2,circle,draw,thick] (3) at (-90:1){$3$};
\draw[red,->,thick] (1.{-120-90-20}) .. controls ({-120-90-20}:2) and ({-120-90+20}:2) .. (1.{-120-90+20});
\draw[red,->,thick] (2.{120-90-20}) .. controls ({120-90-20}:2) and ({120-90+20}:2) .. (2.{120-90+20});
\draw[red,->,thick] (3.{-90-20}) .. controls ({-90-20}:2) and ({-90+20}:2) .. (3.{-90+20});
\path[->,thick]
(1) edge[red,bend left=40] (2)
(2) edge[red,bend right=15] (1)
(2) edge[Green,bend left=20] (1)
(1) edge[Green,bend left=70] (2)
(1) edge[red,bend right=15] (3)
(3) edge[red,bend left=40] (1)
(3) edge[Green,bend left=70] (1)
(1) edge[Green,bend left=20] (3)
(2) edge[red,bend left=40] (3)
(3) edge[red,bend right=15] (2)
(3) edge[Green,bend left=20] (2)
(2) edge[Green,bend left=70] (3)
;
\end{tikzpicture}
\end{array}
\]
\end{example}

With \cref{lem:Delta_A}, we can give new sufficient conditions for $G$ to be trap-separating or converging.  

\begin{proposition}\label{pro:one_positive_cycle}
Suppose that $G$ has a unique positive cycle $C$, and that every negative cycle of $G$ intersects $C$. Then $G$ is trap-separating. If, in addition, $G$ is strong and has at least one negative cycle, then $G$ is converging. 
\end{proposition}

We need the following lemma.

\begin{lemma}[\cite{R18}]\label{lem:R18a}
Suppose that $G$ is strong, has a unique positive cycle, and at least one negative cycle. Then every $f\in F(G)$ has at most one fixed point. 
\end{lemma}

\begin{proof}[\BF{Proof of \cref{pro:one_positive_cycle}}]
Suppose that $G$ has a unique positive cycle $C$, and that every negative cycle of $G$ intersects $C$. Let $f\in F(G)$. We prove that $\Gamma(f)$ is fixing or converging, and thus $G$ is trap-separating. Suppose that $\Gamma(f)$ is not converging. Let $A,B$ be distinct attractors of $\Gamma(f)$. By \cref{lem:Delta_A}, $\Delta(A)$ is disjoint from the vertex set of $C$. Hence, by \cref{lem:R10}, if $|A|\geq 2$ then $G$ has a negative cycle disjoint from $C$, a contradiction. Thus $|A|=1$, that is, $\Gamma(f)$ is fixing. Suppose now that, in addition, $G$ is strong and has at least one negative cycle. If $\Gamma(f)$ is not converging, then it is fixing and thus $f$ has at least two fixed points, and this contradicts \cref{lem:R18a}. Thus $\Gamma(f)$ is always converging. 
\end{proof}

The following example demonstrates that we cannot replace trap-separating with trapping in the first part of \cref{pro:one_positive_cycle}.
\cref{ex:sep-not-conv-fix-graph} shows that we cannot drop the hypothesis of $G$ strong in the second part.

\begin{example}\label{ex:pos-not-trapping}
Consider $f\in F(3)$ defined by $f_1(x)=\bar x_1 x_2$, $f_2(x)=x_1\lor \bar x_2$ and $f_3(x)=x_1 \bar x_2$. $G(f)$ has a unique positive cycle that intersects all cycles. $\Gamma(f)$ is trap-separating but not trapping.
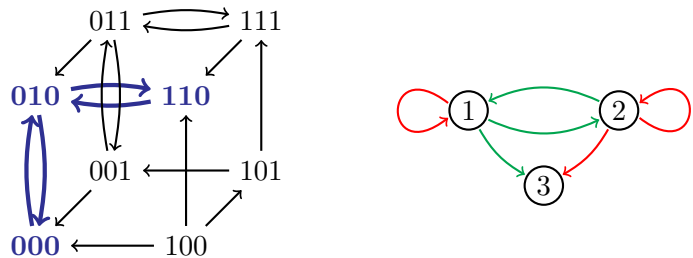
\[
\begin{array}{c}
\begin{tikzpicture}
\pgfmathparse{1}
\node (000) at (0,0){{\boldmath \textcolor{Blue}{$000$} \unboldmath}};
\node (001) at (1,1){$001$};
\node (010) at (0,2){{\boldmath \textcolor{Blue}{$010$} \unboldmath}};
\node (011) at (1,3){$011$};
\node (100) at (2,0){$100$};
\node (101) at (3,1){$101$};
\node (110) at (2,2){{\boldmath \textcolor{Blue}{$110$} \unboldmath}};
\node (111) at (3,3){$111$};
\path[thick,->,draw,black]
(000) edge[bend left=10,ultra thick,Blue] (010)
(010) edge[bend left=10,ultra thick,Blue] (000)
(010) edge[bend left=10,ultra thick,Blue] (110)
(110) edge[bend left=10,ultra thick,Blue] (010)
(001) edge (000)
(001) edge[bend left=10] (011)
(011) edge (010)
(011) edge[bend left=10] (111)
(011) edge[bend left=10] (001)
(100) edge (101)
(100) edge (110)
(100) edge (000)
(101) edge (001)
(101) edge (111)
(111) edge[bend left=10] (011)
(111) edge (110)
;
\end{tikzpicture}
\end{array}
\qquad
\begin{array}{c}
\begin{tikzpicture}
\useasboundingbox (-2.2,-1.3) rectangle (2.2,0.7);
\node[outer sep=1,inner sep=2,circle,draw,thick] (1) at ({180}:1){$1$};
\node[outer sep=1,inner sep=2,circle,draw,thick] (2) at ({0}:1){$2$};
\node[outer sep=1,inner sep=2,circle,draw,thick] (3) at (-90:1){$3$};
\draw[red,->,thick] (1.{180-20}) .. controls ({180-20}:2.3) and ({180+20}:2.3) .. (1.{180+20});
\draw[red,->,thick] (2.{0-20}) .. controls ({0-20}:2.3) and ({0+20}:2.3) .. (2.{0+20});
\path[->,thick]
(1) edge[Green,bend right=25] (2)
(1) edge[Green,bend right=15] (3)
(2) edge[red,bend left=15] (3)
(2) edge[Green,bend right=25] (1)
;
\end{tikzpicture}
\end{array}
\]
\end{example}

What if $G$ is strong and has only one positive cycle? By \cref{thm:PFN} we know that $G$ is separating.
The following example shows that $G$ is not necessarily trapping. It also shows that $G$ is not necessarily trapping if it is strong and has feedback vertex number equal to one.
Whether $G$ is trap-separating in these cases remains an open question.

\begin{example}\label{ex:strong-not-trapping}
Let $f\in F(4)$ be defined by $f_1(x)=\bar x_3$, $f_2(x)=\bar x_1$, $f_3(x)=\bar x_2 \bar x_4$ and $f_4(x)=x_1x_2x_3$. $G(f)$ is strong, has only one positive cycle and has feedback number one, but is not trapping.
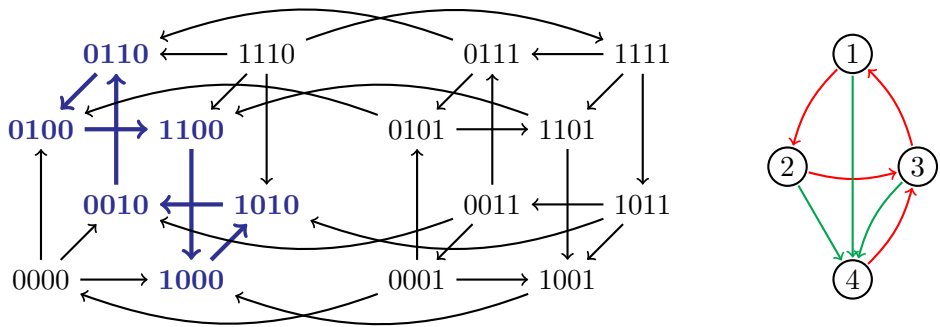
\[
\begin{array}{c}
\begin{tikzpicture}
\pgfmathparse{1}
\node (0000) at (0,0){$0000$};
\node (0010) at (1,1){{\boldmath \textcolor{Blue}{$0010$} \unboldmath}};
\node (0100) at (0,2){{\boldmath \textcolor{Blue}{$0100$} \unboldmath}};
\node (0110) at (1,3){{\boldmath \textcolor{Blue}{$0110$} \unboldmath}};
\node (1000) at (2,0){{\boldmath \textcolor{Blue}{$1000$} \unboldmath}};
\node (1010) at (3,1){{\boldmath \textcolor{Blue}{$1010$} \unboldmath}};
\node (1100) at (2,2){{\boldmath \textcolor{Blue}{$1100$} \unboldmath}};
\node (1110) at (3,3){$1110$};
\node (0001) at (5,0){$0001$};
\node (0011) at (6,1){$0011$};
\node (0101) at (5,2){$0101$};
\node (0111) at (6,3){$0111$};
\node (1001) at (7,0){$1001$};
\node (1011) at (8,1){$1011$};
\node (1101) at (7,2){$1101$};
\node (1111) at (8,3){$1111$};
\path[thick,->,draw,black]
(0100) edge[ultra thick,Blue] (1100)
(1100) edge[ultra thick,Blue] (1000)
(1000) edge[ultra thick,Blue] (1010)
(1010) edge[ultra thick,Blue] (0010)
(0010) edge[ultra thick,Blue] (0110)
(0110) edge[ultra thick,Blue] (0100)
(1110) edge (0110)
(1110) edge (1010)
(1110) edge (1100)
(0000) edge (1000)
(0000) edge (0100)
(0000) edge (0010) 
(0101) edge (1101)
(1101) edge (1001)
(1011) edge (0011)
(1011) edge (1001)
(0011) edge (0001)
(0011) edge (0111)
(0111) edge (0101)
(1111) edge (0111)
(1111) edge (1011)
(1111) edge (1101)
(0001) edge (1001)
(0001) edge (0101)
(0001) edge[bend left=20] (0000)
(0011) edge[bend left=20] (0010)
(0101) edge[bend right=20] (0100)
(0111) edge[bend right=20] (0110)
(1001) edge[bend left=20] (1000)
(1011) edge[bend left=20] (1010)
(1101) edge[bend right=20] (1100)
(1110) edge[bend left=20] (1111)

;
\end{tikzpicture}
\end{array}
\qquad
\begin{array}{c}
\begin{tikzpicture}
\node[outer sep=1,inner sep=2,circle,draw,thick] (1) at (90:1){$1$};
\node[outer sep=1,inner sep=2,circle,draw,thick] (2) at ({90+120}:1){$2$};
\node[outer sep=1,inner sep=2,circle,draw,thick] (3) at ({90-120}:1){$3$};
\node[outer sep=1,inner sep=2,circle,draw,thick] (4) at (-90:2){$4$};
\path[->,thick]
(1) edge[red,bend right=15] (2)
(2) edge[red,bend right=15] (3)
(3) edge[red,bend right=15] (1)
(1) edge[Green] (4)
(2) edge[Green] (4)
(3) edge[Green,bend right=15] (4)
(4) edge[red,bend right=15] (3)
;
\end{tikzpicture}
\end{array}
\]
\end{example}

\section{Number of negative cycles}\label{sec:number-negative-cycles}

In this section, we say more about negative cycles in non-separating signed digraphs. There are non-separating signed digraphs $G$ with negative feedback number one, and even with negative arc-feedback number one (that is, one arc belongs to every negative cycle), as showed by the examples below. 

\begin{example}\label{ex:negative_feedback_1}
Let $f\in F(3)$ be defined by $f_1(x)=x_1+x_2$, $f_2(x)=\bar x_1x_2\lor x_3$ and $f_3(x)=x_1$. Then $\Gamma(f)$ is non-separating and $\{1\}$ is a negative feedback vertex set of $G(f)$. 
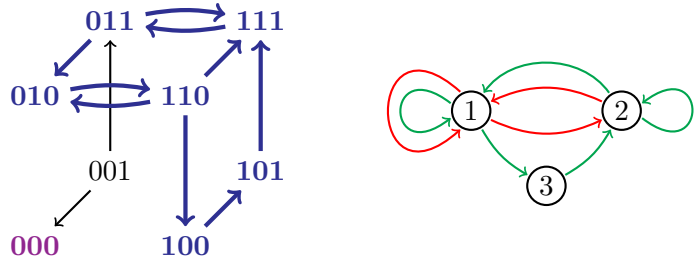
\[
\begin{array}{c}
\begin{tikzpicture}
\pgfmathparse{1}
\node (000) at (0,0){{\boldmath \textcolor{Plum}{$000$} \unboldmath}};
\node (001) at (1,1){$001$};
\node (010) at (0,2){{\boldmath \textcolor{Blue}{$010$} \unboldmath}};
\node (011) at (1,3){{\boldmath \textcolor{Blue}{$011$} \unboldmath}};
\node (100) at (2,0){{\boldmath \textcolor{Blue}{$100$} \unboldmath}};
\node (101) at (3,1){{\boldmath \textcolor{Blue}{$101$} \unboldmath}};
\node (110) at (2,2){{\boldmath \textcolor{Blue}{$110$} \unboldmath}};
\node (111) at (3,3){{\boldmath \textcolor{Blue}{$111$} \unboldmath}};
\path[thick,->,draw,black]
(001) edge (011)
(001) edge (000)
(010) edge[bend left=10,ultra thick,Blue] (110)
(011) edge[bend left=10,ultra thick,Blue] (111)
(011) edge[ultra thick,Blue] (010)
(100) edge[ultra thick,Blue] (101)
(101) edge[ultra thick,Blue] (111)
(110) edge[bend left=10,ultra thick,Blue] (010)
(110) edge[ultra thick,Blue] (100)
(110) edge[ultra thick,Blue] (111)
(111) edge[bend left=10,ultra thick,Blue] (011)
;
\end{tikzpicture}
\end{array}
\qquad
\begin{array}{c}
\begin{tikzpicture}
\useasboundingbox (-2.2,-1.3) rectangle (2.2,0.7);
\node[outer sep=1,inner sep=2,circle,draw,thick] (1) at ({180}:1){$1$};
\node[outer sep=1,inner sep=2,circle,draw,thick] (2) at ({0}:1){$2$};
\node[outer sep=1,inner sep=2,circle,draw,thick] (3) at (-90:1){$3$};
\draw[Green,->,thick] (1.{180-20}) .. controls ({180-20}:2.3) and ({180+20}:2.3) .. (1.{180+20});
\draw[red,->,thick] (1.{180-60}) .. controls ({180-30}:2.8) and ({180+30}:2.8) .. (1.{180+60});
\draw[Green,->,thick] (2.{0-20}) .. controls ({0-20}:2.3) and ({0+20}:2.3) .. (2.{0+20});
\path[->,thick]
(1) edge[red,bend right=25] (2)
(1) edge[Green,bend right=15] (3)
(3) edge[Green,bend right=15] (2)
(2) edge[red,bend right=25] (1)
(2) edge[Green,bend right=55] (1)
;
\end{tikzpicture}
\end{array}
\]
\end{example}

\begin{example}\label{ex:negative_arc_feedback_1}
Let $f\in F(4)$ be defined by $f_1(x)=x_2\bar x_3\lor \bar x_2x_3\lor x_3\bar x_4$, $f_2(x)=x_2\bar x_3\lor x_4$, $f_3(x)=x_1$ and $f_4(x)=x_3$. Then $\Gamma(f)$ is non-separating and every negative cycle contains the positive arc from $1$ to $3$. 
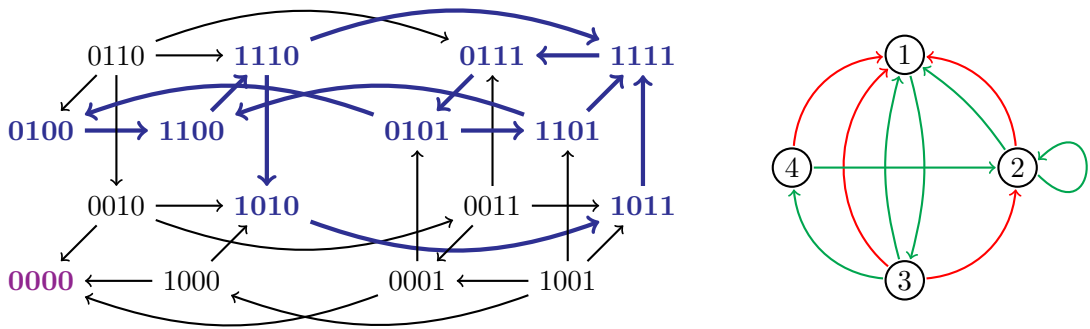
\[
\begin{array}{c}
\begin{tikzpicture}
\pgfmathparse{1}
\node (0000) at (0,0){{\boldmath \textcolor{Plum}{$0000$} \unboldmath}};
\node (0010) at (1,1){$0010$};
\node (0100) at (0,2){{\boldmath \textcolor{Blue}{$0100$} \unboldmath}};
\node (0110) at (1,3){$0110$};
\node (1000) at (2,0){$1000$};
\node (1010) at (3,1){{\boldmath \textcolor{Blue}{$1010$} \unboldmath}};
\node (1100) at (2,2){{\boldmath \textcolor{Blue}{$1100$} \unboldmath}};
\node (1110) at (3,3){{\boldmath \textcolor{Blue}{$1110$} \unboldmath}};
\node (0001) at (5,0){$0001$};
\node (0011) at (6,1){$0011$};
\node (0101) at (5,2){{\boldmath \textcolor{Blue}{$0101$} \unboldmath}};
\node (0111) at (6,3){{\boldmath \textcolor{Blue}{$0111$} \unboldmath}};
\node (1001) at (7,0){$1001$};
\node (1011) at (8,1){{\boldmath \textcolor{Blue}{$1011$} \unboldmath}};
\node (1101) at (7,2){{\boldmath \textcolor{Blue}{$1101$} \unboldmath}};
\node (1111) at (8,3){{\boldmath \textcolor{Blue}{$1111$} \unboldmath}};
\path[thick,->,draw,black]
(0010) edge (1010)
(0010) edge (0000)
(0010) edge[bend right=20] (0011)
(0100) edge[ultra thick,Blue] (1100)
(0110) edge (1110)
(0110) edge (0010)
(0110) edge (0100)
(0110) edge[bend left=20] (0111)
(1000) edge (0000)
(1000) edge (1010)
(1010) edge[ultra thick,bend right=20,Blue] (1011)
(1100) edge[ultra thick,Blue] (1110)
(1110) edge[ultra thick,Blue] (1010)
(1110) edge[ultra thick,bend left=20,Blue] (1111)
(0001) edge (0101)
(0001) edge[bend left=20] (0000)
(0011) edge (1011)
(0011) edge (0111)
(0011) edge (0001)
(0101) edge[ultra thick,Blue] (1101)
(0101) edge[ultra thick,bend right=20,Blue] (0100)
(0111) edge[ultra thick,Blue] (0101)
(1001) edge (0001)
(1001) edge (1101)
(1001) edge (1011)
(1001) edge[bend left=20] (1000)
(1011) edge[ultra thick,Blue] (1111)
(1101) edge[ultra thick,Blue] (1111)
(1101) edge[ultra thick,Blue,bend right=20] (1100)
(1111) edge[ultra thick,Blue] (0111)
;
\end{tikzpicture}
\end{array}
\qquad
\begin{array}{c}
\begin{tikzpicture}
\node[outer sep=1,inner sep=2,circle,draw,thick] (1) at (90:1.5){$1$};
\node[outer sep=1,inner sep=2,circle,draw,thick] (2) at ({0}:1.5){$2$};
\node[outer sep=1,inner sep=2,circle,draw,thick] (4) at ({180}:1.5){$4$};
\node[outer sep=1,inner sep=2,circle,draw,thick] (3) at (-90:1.5){$3$};
\draw[Green,->,thick] (2.{0-20}) .. controls ({0-20}:2.8) and ({0+20}:2.8) .. (2.{0+20});
\path[->,thick]
(1) edge[Green,bend left=15] (3)
(2) edge[Green,bend right=10] (1)
(2) edge[red,bend right=40] (1)
(3) edge[Green,bend left=15] (1)
(3) edge[red,bend left=50] (1)
(3) edge[red, bend right=40] (2)
(3) edge[Green,bend left=40] (4)
(4) edge[Green] (2)
(4) edge[red, bend left=40] (1)
;
\end{tikzpicture}
\end{array}
\]
\end{example}

Our main observation concerning negative cycles is then is the following. 

\begin{theorem}\label{thm:NEGATIVE}
If $G$ has at most one negative cycle, then $G$ is separating. If $G$ is strong and has at most one negative cycle, then $G$ is trapping. 
\end{theorem}

\cref{ex:not-trap-sep} shows that, in the theorem, separating cannot be replaced by trap-separating. For the proof we need some lemmas.

\begin{lemma}\label{lem:almost_fixed_points}
Let $f\in F(G)$. Suppose that $G\setminus i$ has no negative cycle for some vertex $i$ and that $\Gamma(f)$ has an attractor $A$ of size at least two. Then there are $x,y\in A$ with $x_i\neq y_i$ such that $f(x)=x+e_i$ and $f(y)=y+e_i$.
\end{lemma}

\begin{proof}
Let $f'\in F(V)$ be defined by $f'_i(x)=x_i$ and $f'_j(x)_j=f_j(x)$ for all $j\neq i$. Then $\Gamma(f')$ is a spanning subgraph of $\Gamma(f)$, and $G(f')\setminus i=G\setminus i$. So $G(f')$ has no negative cycle. 

\medskip
Let $a\in A$. Since $G(f')$ has no negative cycle, by \cref{thm:bib}, $\Gamma(f')$ is fixing and thus it has a path $P$ from $a$ to a fixed point $x$ of $f'$. By the definition of $f'$, we have $z_i=a_i$ for all the configuration $z$ in $P$. In particular, $x_i=a_i$, and since $a\in A$ and $P$ is a path of $\Gamma(f)$, we have $x\in A$. Since $A$ is of size at least two, $x$ is not a fixed point of $f$, thus $f(x)=x+e_i$. Let $b=x+e_i$. We have $b\in A$, and $b_i\neq a_i$, and we prove similarly, that there is $y$ with $y_i=b_i$ such that $f(y)=y+e_i$.
\end{proof}

\begin{lemma}\label{lem:fixed_point}
Suppose that all the negative arcs of $G$ have the same terminal vertex $i$. Let $f\in F(G)$. If $x_i=0$ and $x\leq f(x)$ then $[x,\ONE]$ is a trap space, and if $x_i=1$ and $x\geq f(x)$ then $[\ZERO,x]$ is a trap space. 
\end{lemma}

\begin{proof}
Suppose that $x_i=0$ and $x\leq f(x)$. Let $z\in [x,\ONE]$. Then $0=x_i\leq f_i(z)$ and for all $j\neq i$ we have $x_j\leq f_j(x)\leq f_j(z)$ since $f_j$ is monotone and $x\leq z$. Thus $x\leq f(z)$ and we deduce that $[x,\ONE]$ is a trap space. We prove similarly that $[\ZERO,x]$ is a trap space if $x_i=1$ and $x\geq f(x)$.
\end{proof}

\begin{lemma}\label{lem:trap_spaces}
Suppose that all the negative arcs of $G$ have the same terminal vertex $i$. Let $f\in F(G)$ and $A$ an attractor of $\Gamma(f)$ of size at least two. There are $x,y\in A$ with $x_i<y_i$ and $x\leq y$ such that $[A]=[x,y]$. Furthermore, $[A]$, $[x,\ONE]$ and $[\ZERO,y]$ are trap spaces.
\end{lemma}

\begin{proof}
By \cref{lem:almost_fixed_points}, there are $x,y\in A$ with $x_i<y_i$ such that $f(x)=x+e_i$ and $f(y)=y+e_i$, and thus $x\leq f(x)$ and $y\geq f(y)$. By \cref{lem:fixed_point}, $[x,\ONE]$ and $[\ZERO,y]$ are trap spaces. It follows that $x\leq y$ and that $[A]=[x,y]$ is a trap space.
\end{proof}

\begin{lemma}\label{lem:special_vertex}
If all the negative arcs of $G$ have the same terminal vertex, then $G$ is trapping.
\end{lemma}

\begin{proof}
Let $i$ be the terminal vertex of every negative arc of $G$. Let $f\in F(G)$ and suppose that $A,B$ are distinct attractors of $\Gamma(f)$. By \cref{lem:trap_spaces}, $[A]$ and $[B]$ are trap spaces, so it remains to prove that $[A]\cap [B]=\emptyset$. Suppose, for a contradiction, that $[A]\cap [B]\neq\emptyset$. Then at least one of $A,B$ is of size at least two, say $A$. By \cref{lem:trap_spaces}, there are $x^A,y^A\in A$ with $x^A_i<y^A_i$ such that $[x^A,\ONE]$ and $[A]=[x^A,y^A]$ are trap spaces.

\medskip
Suppose first that $B$ is of size one, that is, consists of a fixed point, say $z$. Then $z\in [x^A,y^A]$. If $z_i=0$ then, by \cref{lem:fixed_point}, $[z,\ONE]$ is a trap space. We have $y^A\in [z,\ONE]$ and since $z\neq x^A$, we have $x^A\not\in [z,\ONE]$. We deduce that there is no path in $\Gamma(f)$ from $y^A$ to $x^A$, a contradiction. If $z_i=1$ we obtain a contradiction similarly.

\medskip
Consequently, $B$ is of size at least two. Hence, by \cref{lem:trap_spaces}, there are $x^B,y^B\in B$ with $x^B_i<y^B_i$ such that $[x^B,\ONE]$ and $[B]=[x^B,y^B]$ are trap spaces. If $y^B_j<x^A_j$ for some vertex $j$, then $[A]\cap [B]=\emptyset$, a contradiction. So $x^A\leq y^B$. Hence $y^B\in [x^A,\ONE]$ and since $[x^A,\ONE]$ is a trap space, we deduce that $B\subseteq [x^A,\ONE]$. In particular, $x^A\leq x^B$. By symmetry we have $x^B\leq x^A$. Thus $x^A=x^B$, a contradiction.
\end{proof}

\begin{lemma}\label{lem:switch}
If $G$ is strong and contains an arc from $j$ to $i$, or a vertex $i$, that belongs to every negative cycle and to no positive cycle, then, up to a switch of $G$, all the negative arcs have $i$ as terminal vertex.
\end{lemma}

\begin{proof}
Let $G'$ be obtained from $G$ by deleting all the in-coming arcs of $i$. Since $G$ is strong, $G'$ has a unique initial strong component, which has $i$ has unique vertex, and $G$ has an arc from each terminal strong component of $G'$ to $i$. Since each strong component of $G'$ only contains positive cycles, up to a switch we can assume that all the strong components of $G'$ only contain positive arcs (using \cref{pro:harary,pro:BN_switch}). Let $I_1,\dots,I_r$ be the vertex sets of the strong components of $G'$ in the topological order (so $I_1=\{i\}$). Let $H$ be the signed digraph on $\{I_1,\dots,I_r\}$ with a positive (negative) arc from $I_p$ to $I_q$ if $G$ has a positive (negative) arc from some vertex in $I_p$ to some vertex in $I_q$. Let $s_1=1$ and, for $1<p\leq r$, let $s_p=1$ if $H$ has a positive path from $I_1$ to $I_p$ and $s_p=-1$ otherwise. A consequence of (1) below is that, actually, all the paths from $I_1$ to $I_p$ have the same sign.

\begin{quote}
(1) {\em For $1\leq p<q\leq r$, the sign of an arc of $H$ from $I_p$ to $I_q$ is $s_p\cdot s_q$.}

\smallskip
Suppose that $H$ has an arc from $I_p$ to $I_q$ of sign $s\neq s_p\cdot s_q$. Since $H$ has a path from $I_q$ to $I_1$ with internal vertices in $\{I_{q+1},\dots,I_r\}$, $G$ has a path $P$ from some $k\in I_q$ to $i$ whose internal vertices are in $I_{q+1}\cup\dots \cup I_r$. Since $s\neq s_p\cdot s_q$, $H$ has a positive and a negative path from $I_1$ to $I_q$ whose vertices are in $\{I_1,\dots,I_q\}$. Since $I_q$ is strong and has only positive arcs, we deduce that $G$ has a positive path $P^+$ from $i$ to $k$ and a negative path $P^-$ from $i$ to $k$ whose internal vertices are in $I_2\cup\cdots\cup I_q$. Thus $C_1=P^+\cup P$ and $C_2=P^-\cup P$ are cycles with different signs.

So $i$ belongs to cycles of different signs. We deduce that $i$ has an in-neighbor $j$ such that the arc from $j$ to $i$ belongs to all the negative cycles and to no positive cycle. Thus this arc belongs to $C_1$ or $C_2$, and this means that it is in $P$, and thus it belongs to both $C_1$ and $C_2$, and we obtain a contradiction.
\end{quote}

\medskip
Let $J$ be the set of vertices $I_p$ of $H$ with $s_p=-1$. We deduce from (1) that an arc of $H$ is negative if and only if it leaves $J$ or enters $J$. Hence the $J$-switch of $H$ is full-positive. Let $J'$ be the union of the sets $I_p$ contained in $J$. Then the $J'$-switch of $G'$ is full-positive, so, in the $J'$-switch of $G$, all the negative arcs have $i$ as terminal vertex.
\end{proof}

As a consequence:

\begin{lemma}\label{lem:special_vertex_bis}
If $G$ has an arc or a vertex that belongs to all the negative cycles and to no positive cycle, then $G$ is separating.
\end{lemma}

\begin{proof}
Suppose that $G$ is a smallest counter example with respect to the number of vertices. There is $f\in F(G)$ and attractors $A,B$ of $\Gamma(f)$ with $[A]\cap[B]\neq\emptyset$.

\begin{quote}
(1) {\em For every $j\in V$ there is $a\in A$ and $b\in B$ with $a_j\neq b_j$.}

\smallskip
Suppose that there is $j\in V$ and $c\in\B$ such that $a_j=b_j=c$ for all $a\in A$ and $b\in B$. Let $h$ be the BN  with component set $I=V\setminus j$ defined by $h(x_I)=f(x)_I$ for all $x$ with $x_j=c$. Then $G(h)$ is a subgraph of $G\setminus j$. Let $A'=\{a_I\mid a\in A\}$ and $B'=\{b_I\mid b\in B\}$. Then $A',B'$ are distinct attractors of $\Gamma(h)$ with $[A']\cap [B']\neq\emptyset$. Hence $G(h)$ is non-separating. So, by \cref{thm:bib}, $G(h)$ has a negative cycle $C$. Hence $C$ contains an arc or a vertex that belongs to all the negative cycles of $G$ and to no positive cycle of $G$. Since $G(h)$ is a subgraph of $G\setminus j$, this arc or this vertex belongs to all the negative cycles of $G(h)$ and to no positive cycle of $G(h)$. We deduce that $G(h)$ is a smaller counter example, a contradiction.
\end{quote}

\begin{quote}
(2) {\em $G$ is strong.}

\smallskip
If not there is a partition $(I_1,I_2)$ of the vertices such that $G$ has no arc from $I_2$ to $I_1$ and $G[I_2]$ is strong. We then use the notations of \cref{lem:decomposition}. Let $i$ be a vertex of $G$ meeting every negative cycle, which exists by hypothesis.

If $i\in I_1$, then $G[I_1]$ is separating, since otherwise it is a smaller counter example. Thus $\Gamma^1$ is separating. We then deduce from \cref{lem:decomposition} that if $A_1\neq B_1$ then $[A_1]\cap [B_1]=\emptyset$ and thus $[A]\cap [B]=\emptyset$, a contradiction. So suppose that $A_1=B_1$, which implies $\Gamma^2_A=\Gamma^2_B$ and $A_2\neq B_2$. Since $i$ is in $I_1$, $G[I_2]$ has only positive cycles, so by \cref{lem:robustly_trapping} it is robustly trapping, and thus $\Gamma^2_A$ is separating. We then deduce from \cref{lem:decomposition} that $[A_2]\cap [B_2]=\emptyset$ and thus $[A]\cap [B]=\emptyset$, a contradiction.

If $i\not\in I_1$ then $G[I_1]$ has only positive cycles. So $G[I_1]$ is fixing, and we deduce from \cref{lem:decomposition} that there are $x^A,x^B$, fixed points of $f^1$, such that $A=\{x^A\}\times A^2$ and $B=\{x^B\}\times B^2$. If $x^A\neq x^B$ then $[A]\cap [B]=\emptyset$, a contradiction. Thus $x^A=x^B$ so $a_{I_1}=b_{I_1}$ for all $a\in A$ and $b\in B$, which contradicts~(1). This proves (2).
\end{quote}

From (2) and \cref{lem:switch}, in some switch $G'$ of $G$, all the negative arcs have the same terminal vertex. By \cref{lem:special_vertex}, $G'$ is separating and we deduce that $G$ is separating by \cref{pro:BN_switch}, a contradiction.
\end{proof}

\begin{lemma}[\cite{R18}]\label{lem:R18}
If $G$ has a unique negative cycle, then some arc belongs to no positive~cycle.
\end{lemma}

\begin{proof}[\BF{Proof of \cref{thm:NEGATIVE}}]
Suppose that $G$ has a unique negative cycle $C$. By \cref{lem:R18}, $C$ has an arc that belongs to no positive cycle. Since this arc obviously belongs to all the negative cycles of $G$, we deduce from \cref{lem:special_vertex_bis} that $G$ is separating. Suppose, in addition, that $G$ is strong. By \cref{lem:switch}, in some switch $G'$ of $G$, all the negative arcs have the same terminal vertex. Hence, by \cref{lem:special_vertex}, $G'$ is trapping and, by \cref{pro:BN_switch}, $G$ is trapping.
\end{proof}

Using the previous tools, we provide a new sufficient condition for $G$ to be fixing.

\begin{proposition}\label{pro:fixing}
If $G$ is strong, has a unique negative cycle $C$, at least one positive cycle, and if no cycle of $G$ is disjoint from $C$, then $G$ is fixing.
\end{proposition}

We need the following lemma.

\begin{lemma}\label{lem:one_negative}
Suppose that $G$ has a unique negative arc, say from $j$ to $i$, where $i$ is of in-degree at least two, and suppose that every positive cycle intersects every negative cycle. Then $G$ is~fixing.
\end{lemma}

\begin{proof}
Let $f\in F(G)$ and suppose, for a contradiction, that $\Gamma(f)$ has an attractor $A$ of size at least two. By \cref{lem:trap_spaces}, there are $x,y\in A$ with $x_i<y_i$ and $x\leq y$ such that $[x,\ONE]$, $[\ZERO,y]$ and $[A]=[x,y]$ are trap spaces. By \cref{lem:R10}, $G$ has a negative cycle $C$ such that $x_k<y_k$ for every vertex $k$ in $C$.

\medskip
Suppose that $x=\ZERO$ and $y=\ONE$. Then, for every vertex $k\neq i$ in $G$, $k$ has only positive in-neighbors. Furthermore, since $x,y\in A$, $\Gamma(f)$ has a path from $\ZERO$ to $\ONE$ and thus $f_k$ is not a constant. So we have $f_k(x)=0=x_k$ and $f_k(y)=1=y_k$. We deduce that $f_i(x)=1$ and $f_i(y)=0$ (otherwise $x$ or $y$ would be fixed points, which is not possible since $x,y\in A$). Let $z$ be any configuration on $V$. If $z_j=0$ then we have $x_j=z_j$ and $x\leq z$, and since all the in-neighbors $k\neq j$ of $i$ are positive, we deduce that $1=f_i(x)\leq f_i(z)$, thus $f_i(z)=1$. Similarly, if $z_j=1$ then we have $z_j=y_j$ and $z\leq y$, and since all the in-neighbors $k\neq j$ of $i$ are positive, we deduce that $f_i(z)\leq f_i(y)=0$, thus $f_i(z)=0$. Consequently, $f_i(z)=z_j+1$ for all configurations $z$. But this means that $j$ is the unique in-neighbor of $i$, a contradiction. Consequently, $x\neq\ZERO$ or $y\neq\ONE$.

\medskip
Suppose first that $x\neq\ZERO$, that is,  $I=\{k\mid x_k=1\}$ is non-empty. Let $k\in I$. Since $[x,\ONE]$ is a trap space, we have $x\leq f(x)$ thus $f_k(x)=1$. Since $k\neq i$ (because $x_i=0$), all the in-coming arcs of $k$ are positive, so it has an in-neighbor $\ell$ with $x_\ell=1$. Hence $\ell\in I$. Consequently, $G[I]$ has a cycle, which is positive since $i\not\in I$. Since $x_k=0$ for all vertices $k$ in $C$, this positive cycle is disjoint from $C$, a contradiction. If $y\neq\ONE$ we prove with similar arguments that $\{k\mid y_k=0\}$ induces a positive cycle disjoint from $C$.
\end{proof}

\begin{proof}[\BF{Proof of \cref{pro:fixing}}]
By \cref{lem:R18}, the unique negative cycle $C$ has an arc that belongs to no positive cycle. Since some positive cycle intersects $C$, some vertex in $C$ has in-degree at least two, and we deduce that $C$ has an arc $a$, say from $j$ to $i$, which belongs to no positive cycle and such that $i$ is of in-degree at least two. Hence, since $G$ is strong, by \cref{lem:switch}, in some switch $G'$ of $G$, all the negative arcs have $i$ as terminal vertex. Let $a'\neq a$ be an arc with terminal vertex $i$, and let $C'$ be a cycle of $G'$ containing $a'$, which exists since $G$ is strong. In $G'$, all the arcs of $C'$ distinct from $a'$ are positive, and if $a'$ is negative, then $C'$ is a negative cycle distinct from $C$, a contradiction. Thus $a$ is the unique negative arc of $G'$. Since $i$ is of in-degree at least two, by \cref{lem:one_negative}, $G'$ is fixing, and so is $G$.
\end{proof}

\cref{ex:sep-not-conv-fix-graph} shows that we cannot drop the hypothesis of $G$ strong in \cref{pro:fixing}. If $G$ is strong, has a unique negative cycle $C$, at least one positive cycle, but some positive cycle is disjoint from $C$ then $G$ is not necessarily fixing, as showed by \cref{ex:fix1} below. Furthermore, if $G$ is strong, has two negative cycles and every positive cycle intersect every negative cycle, then $G$ is not necessarily fixing, as showed by \cref{ex:fix2} below.

\begin{example}\label{ex:fix1}
Let $f\in F(2)$ be defined by $f_1(x)=\bar x_1\lor x_2$ and $f_2(x)=x_1 x_2$. Then $\Gamma(f)$ is not fixing since $\{00,10\}$ is an attractor, while $G(f)$ is strong, has a unique negative cycle and two positive cycles.
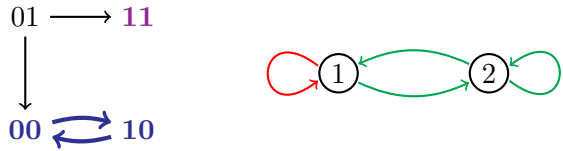
\[
\begin{array}{c}
\begin{tikzpicture}
\pgfmathparse{1}
\node (00) at (0,0){{\boldmath \textcolor{Blue}{$00$} \unboldmath}};
\node (01) at (0,1.5){$01$};
\node (10) at (1.5,0){{\boldmath \textcolor{Blue}{$10$} \unboldmath}};
\node (11) at (1.5,1.5){{\boldmath \textcolor{Plum}{$11$}\unboldmath}};
\path[thick,->,draw,black]
(00) edge[ultra thick,bend left=15,Blue] (10)
(01) edge (11)
(01) edge (00)
(10) edge[ultra thick,bend left=15,Blue] (00)
;
\end{tikzpicture}
\end{array}
\qquad
\begin{array}{c}
\begin{tikzpicture}
\node[outer sep=1,inner sep=2,circle,draw,thick] (1) at ({180}:1){$1$};
\node[outer sep=1,inner sep=2,circle,draw,thick] (2) at ({0}:1){$2$};
\draw[red,->,thick] (1.{180-20}) .. controls ({180-20}:2.3) and ({180+20}:2.3) .. (1.{180+20});
\draw[Green,->,thick] (2.{0-20}) .. controls ({0-20}:2.3) and ({0+20}:2.3) .. (2.{0+20});
\path[->,thick]
(1) edge[Green,bend right=25] (2)
(2) edge[Green,bend right=25] (1)
;
\end{tikzpicture}
\end{array}
\]
\end{example}

\begin{example}\label{ex:fix2}
Let $f\in F(2)$ be defined by $f_1(x)=\bar x_1\lor  x_2$ and $f_2(x)=x_1 \bar x_2$. Then $\Gamma(f)$ is not fixing since $\{00,10,01\}$ is an attractor, while $G(f)$ has a unique positive cycle, which intersects the two negative cycles.
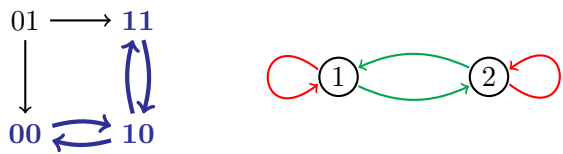
\[
\begin{array}{c}
\begin{tikzpicture}
\pgfmathparse{1}
\node (00) at (0,0){{\boldmath \textcolor{Blue}{$00$} \unboldmath}};
\node (01) at (0,1.5){$01$};
\node (10) at (1.5,0){{\boldmath \textcolor{Blue}{$10$} \unboldmath}};
\node (11) at (1.5,1.5){{\boldmath \textcolor{Blue}{$11$} \unboldmath}};
\path[thick,->,draw,black]
(00) edge[ultra thick,bend left=15,Blue] (10)
(10) edge[ultra thick,bend left=15,Blue] (00)
(10) edge[ultra thick,bend left=15,Blue] (11)
(11) edge[ultra thick,bend left=15,Blue] (10)
(01) edge (11)
(01) edge (00)
;
\end{tikzpicture}
\end{array}
\qquad
\begin{array}{c}
\begin{tikzpicture}
\node[outer sep=1,inner sep=2,circle,draw,thick] (1) at ({180}:1){$1$};
\node[outer sep=1,inner sep=2,circle,draw,thick] (2) at ({0}:1){$2$};
\draw[red,->,thick] (1.{180-20}) .. controls ({180-20}:2.3) and ({180+20}:2.3) .. (1.{180+20});
\draw[red,->,thick] (2.{0-20}) .. controls ({0-20}:2.3) and ({0+20}:2.3) .. (2.{0+20});
\path[->,thick]
(1) edge[Green,bend right=25] (2)
(2) edge[Green,bend right=25] (1)
;
\end{tikzpicture}
\end{array}
\]
\end{example}

\section{Non-separating signed digraphs with feedback number two}\label{sec:feedback-number-two}

We can say more on non-separating signed digraphs $G$ when the feedback number of $G$ is exactly two. Let $K^\pm_n$ be the signed digraph with vertex set $[n]$ and with both a positive and a negative arc from $i$ to $j$ for any $i,j\in [n]$. It is an easy exercise to prove that $K^\pm_2$ is the unique non-separating signed digraph on two vertices, and that $F(K^\pm_2)$ is {\em exactly} the set of BNs in $F(2)$ with a non-separating asynchronous graph. An example follows.

\begin{example}\label{ex:K2}
Let $f\in F(K^\pm_2)$ be defined by $f_1(x)=x_1+x_2$ and $f_2(x)=x_1+x_2$. $\Gamma(f)$ has two attractors: $A=\{00\}$ and $B=\{01,10,11\}$. Since $[B]=\B^2$, $\Gamma(f)$ is non-separating.
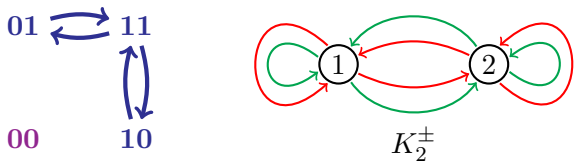
\[
\begin{array}{c}
\begin{tikzpicture}
\pgfmathparse{1}
\node (00) at (0,0){{\boldmath \textcolor{Plum}{$00$} \unboldmath}};
\node (01) at (0,1.5){{\boldmath \textcolor{Blue}{$01$} \unboldmath}};
\node (10) at (1.5,0){{\boldmath \textcolor{Blue}{$10$} \unboldmath}};
\node (11) at (1.5,1.5){{\boldmath \textcolor{Blue}{$11$} \unboldmath}};
\path[thick,->,draw,black]
(01) edge[ultra thick,bend left=15,Blue] (11)
(10) edge[ultra thick,bend left=15,Blue] (11)
(11) edge[ultra thick,bend left=15,Blue] (10)
(11) edge[ultra thick,bend left=15,Blue] (01)
;
\end{tikzpicture}
\end{array}
\qquad
\begin{array}{c}
\begin{tikzpicture}
\useasboundingbox (-2.2,-0.7) rectangle (2.2,0.7);
\node[outer sep=1,inner sep=2,circle,draw,thick] (1) at ({180}:1){$1$};
\node[outer sep=1,inner sep=2,circle,draw,thick] (2) at ({0}:1){$2$};
\draw[Green,->,thick] (1.{180-20}) .. controls ({180-20}:2.3) and ({180+20}:2.3) .. (1.{180+20});
\draw[red,->,thick] (1.{180-60}) .. controls ({180-30}:2.8) and ({180+30}:2.8) .. (1.{180+60});
\draw[Green,->,thick] (2.{0-20}) .. controls ({0-20}:2.3) and ({0+20}:2.3) .. (2.{0+20});
\draw[red,->,thick] (2.{0-60}) .. controls ({0-30}:2.8) and ({0+30}:2.8) .. (2.{0+60});
\path[->,thick]
(1) edge[red,bend right=25] (2)
(1) edge[Green,bend right=55] (2)
(2) edge[red,bend right=25] (1)
(2) edge[Green,bend right=55] (1)
;
\end{tikzpicture}
\\
K^\pm_2
\end{array}
\]
\end{example}

Let $H_2$ be obtained from $K^\pm_2$ by deleting the negative loop on vertex $2$.
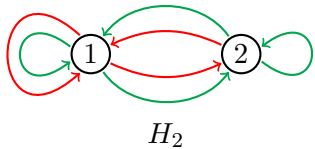
\[
\begin{array}{c}
\begin{tikzpicture}
\useasboundingbox (-2.2,-0.7) rectangle (2.2,0.7);
\node[outer sep=1,inner sep=2,circle,draw,thick] (1) at ({180}:1){$1$};
\node[outer sep=1,inner sep=2,circle,draw,thick] (2) at ({0}:1){$2$};
\draw[Green,->,thick] (1.{180-20}) .. controls ({180-20}:2.3) and ({180+20}:2.3) .. (1.{180+20});
\draw[red,->,thick] (1.{180-60}) .. controls ({180-30}:2.8) and ({180+30}:2.8) .. (1.{180+60});
\draw[Green,->,thick] (2.{0-20}) .. controls ({0-20}:2.3) and ({0+20}:2.3) .. (2.{0+20});
\path[->,thick]
(1) edge[red,bend right=25] (2)
(1) edge[Green,bend right=55] (2)
(2) edge[red,bend right=25] (1)
(2) edge[Green,bend right=55] (1)
;
\end{tikzpicture}
\\
H_2
\end{array}
\]

We prove below that if $G$ is non-separating and has feedback number two, then $G$ contains in some way $H_2$. To make this precise we need some definitions. Let $H$ be a signed digraphs with vertex set $U$. We say that $H$ is \EM{embedded} in $G$ if there is an injection $\phi:U\to V$ such that, for every positive (negative) arc of $H$ from $j$ to $i$, $G$ contains a positive (negative) path from $\phi(j)$ to $\phi(i)$ whose internal vertices are not in $\phi(U)$.

\begin{theorem}\label{thm:fvs2}
If $G$ is non-separating and has feedback number $2$, then $H_2$ is embedded in $G$.
\end{theorem}

\begin{remark}
The proof also shows that if $G$ has feedback number $2$ and $\Gamma(f)$ is non-separating for some $f\in F(G)$, then $\Gamma(f)$ has exactly two attractors, say $A$ and $B$, with $|A|=1$ and $|B|\geq 3$.
\end{remark}

Note that if $G$ is non-separating and has feedback number $2$, then $K^\pm_2$ is not necessarily embedded in $G$, as illustrated by \cref{ex:negative_feedback_1}.

\begin{example}
The interaction graph $G$ of \cref{ex:negative_arc_feedback_1} is non-separating and has feedback number~$2$. $H_2$ is indeed embedded in $G$, with $\phi(1)=1$ and $\phi(2)=2$, because: $1~\textcolor{Green}{\to}~3~\textcolor{Green}{\to}~1$ and $1~\textcolor{Green}{\to}~3~\textcolor{red}{\to}~1$ are positive and negative cycles containing $1$ but not $2$; $1~\textcolor{Green}{\to}~3~\textcolor{red}{\to}~2$ and $1~\textcolor{Green}{\to}~3~\textcolor{Green}{\to}~4~\textcolor{Green}{\to}~2$ are positive and negative paths from $1$ to $2$; $2~\textcolor{Green}{\to}~1$ and $2~\textcolor{red}{\to}~1$ are positive and negative paths from $2$ to $1$; and $2~\textcolor{Green}{\to}~2$ is a positive cycle containing $2$ but not $1$.
\end{example}

We will use several times a lemma whose statement needs some definitions. Let $P$ be a path of an asynchronous graph $\Gamma$ of length $\ell\geq 1$, with configurations  $x^0,\dots,x^\ell$ in order. For $0\leq k<\ell$, let $i_k$ be the direction of the arc $x^k\to x^{k+1}$. We call $i_0,\dots,i_{\ell-1}$ the \EM{direction sequence} of $P$. We say that $P$ is a \EM{geodesic} if its direction sequence has no repetition. We say that $P$ is \EM{increasing} if all its arcs are increasing. A \EM{walk} $W$ in a signed digraph $G$ is a sequence of arcs $a^0,\dots,a^\ell$ such that for $0\leq k<\ell$, the terminal vertex of $a^k$ is the initial vertex of $a^{k+1}$. The sign of $W$ is the product of the sign of its arcs. For $0\leq k<\ell$, let $i_k$ be the initial vertex of $a^k$, and let $i_\ell$ be terminal vertex of $a^\ell$. Then $i_0,\dots,i_\ell$ is the \EM{vertex sequence} of $W$, and we say that $W$ is a walk from $i_0$ to $i_\ell$. A sequence $s'$ is a \EM{subsequence} of $s$ if we can obtain $s'$ by removing some elements of~$s$ (so for instance $ii$ is a subsequence of $iji$).

\begin{lemma}[\cite{R10}]\label{lem:R10b}
Let $f\in F(G)$, and let $P$ be a path of $\Gamma(f)$ of length $\ell\geq 2$, with configurations  $x^1x^2\dots x^\ell x^{\ell+1}$ in order. Let $i$ be the direction of the arc from $x^\ell$ to $x^{\ell+1}$, and suppose that $f_i(x^k)=x^k_i$ for $1\leq k<\ell$. There is a component $j$ with $f_j(x^1)\neq x^1_j$ such that $G$ has a walk $W$ from $j$ to $i$ of sign $(f_j(x^1)-x^1_j)(f_i(x^\ell)-x^\ell_i)$ such that the vertex sequence of $W$ is a subsequence of the direction sequence of $P$. Furthermore, if $P$ is increasing then $W$ is a full-positive path (its vertex sequence has no repetition, and all its arcs are positive) and $j\in\Delta(x^1,x^{\ell+1})$.
\end{lemma}

The following is a well-known result of Robert.

\begin{lemma}[\cite{R95}]\label{lem:R95}
Suppose that $G$ is acyclic and let $f\in F(G)$. Then $f$ has a unique fixed point and $\Gamma(f)$ has a geodesic from any configuration on $V$ to this fixed point.
\end{lemma}

\begin{proof}[\BF{Proof of \cref{thm:fvs2}}]
Let $\{i_1,i_2\}$ be a feedback vertex set of $G$. Let $f\in F(G)$, and let $A,B$ be distinct attractors of $\Gamma=\Gamma(f)$ with $[A]\cap [B]\neq\emptyset$. We will prove that $H_2$ is embedded in $G$. For $c_1,c_2\in\B$, let $X^{c_1c_2}$ be the set of configurations on $V$ with $x_{i_1}=c_1$ and $x_{i_2}=c_2$.

\begin{quote}
(1) {\em For every $c_1,c_2\in \B$, there is $x^{c_1c_2}\in X^{c_1c_2}$ with $f_j(x^{c_1c_2})=x^{c_1c_2}_j$ for all $j\neq i_1,i_2$ such that $\Gamma$ has a geodesic from any configuration in $X^{c_1c_2}$ to $x^{c_1c_2}$.}

\smallskip
Let $f'\in F(V)$ be defined by $f'_{i_1}(x)=c_1$, $f'_{i_2}(x)=c_2$ and $f'_j(x)=f_j(x)$ for all $j\neq i_1,i_2$. Then $G(f')$ is obtained from $G$ by removing all the arcs with $i_1$ or $i_2$ as terminal vertex, and thus it is acyclic. By \cref{lem:R95}, $f'$ has a unique fixed point, say $x^{c_1c_2}$, and $\Gamma(f')$ has a geodesic from any configuration to $x^{c_1c_2}$. Since $x^{c_1c_2}\in X^{c_1c_2}$, and since $\Gamma(f')[X^{c_1c_2}]=\Gamma[X^{c_1c_2}]$, we deduce that $x^{c_1c_2}$ has the desired properties.
\end{quote}

We deduce from (1) that $A$ and $B$ cannot intersect the same set $X^{c_1c_2}$ (since otherwise $x^{c_1c_2}\in A\cap B$). Suppose that $A$ and $B$ intersect at most two of the four sets $X^{00},X^{01},X^{10},X^{11}$, and that $A$ intersects $X^{c_1c_2}$. If $A$ intersects also $X^{\bar c_1\bar c_2}$ then it must intersects $X^{\bar c_1 c_2}$ or $X^{c_1 \bar c_2}$, and we obtain a contradiction. We deduce that either $A\subseteq X^{c_1c_2}\cup X^{\bar c_1c_2}$ and $B\subseteq X^{\bar c_1\bar c_2}\cup X^{c_1\bar c_2}$ or $A\subseteq X^{c_1c_2}\cup X^{c_1\bar c_2}$ and $B\subseteq X^{\bar c_1\bar c_2}\cup X^{\bar c_1c_2}$. In both cases we have $[A]\cap [B]=\emptyset$. Hence we can suppose that one of $A,B$  intersects three of the sets, and the other one. Up to a switch of $f$, we can suppose that:
\[
x^{00}=\ZERO, \quad B\subseteq X^{00}\quad\textrm{and}\quad
A\subseteq X^{01}\cup X^{10}\cup X^{11}.
\]
From (1) we have $\ZERO\in B$ and $x^{01},x^{10},x^{11}\in A$.

\begin{quote}
(2) {\em $f(\ZERO)=\ZERO$, $f(x^{01})=x^{01}+e_{i_1}$ and $f(x^{10})=x^{10}+e_{i_2}$.}

\smallskip
If $f_{i_1}(\ZERO)=1$ then $\Gamma$ has an arc from $\ZERO$ to $e_{i_1}$ and we deduce that $e_{i_1}\in B\cap X^{10}$, a contradiction. Thus $f_{i_1}(\ZERO)=0$ and we prove similarly that $f_{i_2}(\ZERO)=0$. We deduce from (1) that $f(\ZERO)=\ZERO$. If $f_{i_2}(x^{01})=0$, then $\Gamma$ has an arc from $x^{01}$ to $x^{01}+e_{i_2}$ and we deduce that $x^{01}+e_{i_2}\in A\cap X^{00}$, a contradiction. Thus $f_{i_2}(x^{01})=1$, and since $x^{01}$ is not a fixed point, we deduce from (1) that $f(x^{01})=x^{01}+e_{i_1}$. We prove similarly that $f(x^{10})=x^{10}+e_{i_2}$.
\end{quote}

\begin{quote}
(3) {\em $G\setminus i_2$ has a full-positive cycle containing $i_1$, and $G\setminus i_1$ has a full-positive cycle containing $i_2$.}

\smallskip
Let $f'\in F(V)$ be defined by $f'_{i_2}(x)=0$ and $f'_j(x)=f_j(x)$ for all $j\neq i_2$. Then $G(f')\setminus i_2=G\setminus i_2$. By (2), we have $f'(x^{10})=x^{10}$. Since $f'(\ZERO)=\ZERO$ and $x^{10}_{i_2}=0$, we deduce from \cref{lem:A08} that $G(f')\setminus i_2=G\setminus i_2$ has a full-positive cycle~$C$. Since $i_1$ is a feedback vertex set of $G\setminus i_2$, we deduce that $C$ contains $i_1$. We prove similarly that $G\setminus i_1$ has a full-positive cycle containing $i_2$.
\end{quote}

\begin{quote}
(4) {\em $G$ has a full-positive path from $i_1$ to $i_2$, and from $i_2$ to $i_1$.}

\smallskip
We prove that $G$ has a full-positive path from $i_1$ to $i_2$; for the other path the argument is similar. If $f_{i_2}(e_{i_1})=1$ then, since $f_{i_2}(\ZERO)=0$, $G$ has a positive arc from $i_1$ to $i_2$ and we are done. Suppose $f_{i_2}(e_{i_1})=0$. Let $P$ be a geodesic of $\Gamma$ from $e_{i_1}$ to $x^{10}$ which exists by (1), and which is increasing since $e_{i_1}\leq x^{10}$. Let $x$ be the first vertex of $P$ with $f_{i_2}(x)=1$, which exists by (2), and which is not the first vertex of $P$, by hypothesis. Let $P'$ be obtained by adding the arc from $x$ to $x+e_{i_2}$ to the subpath of $P$ from $e_{i_1}$ to $x$. Then $P'$ is increasing and, by \cref{lem:R10b}, there is $j\neq i_1,i_2$ such that $f_j(e_{i_1})\neq (e_{i_1})_j$ and such that $G$ has a full-positive path from $j$ to $i_2$ (which does not contain $i_1$). So $f_j(e_{i_1})=1$ and, since $f_j(\ZERO)=0$, $G$ has a positive arc from $i_1$ to $j$, and thus $G$ has a full-positive path from $i_1$ to $i_2$.
\end{quote}

\begin{quote}
(5) {\em If $f_{i_1}(x^{11})=0$ then $G$ has a negative cycle containing $i_1$ but not $i_2$, and if $f_{i_2}(x^{11})=0$ then $G$ has a negative cycle containing $i_2$ but not $i_1$.}

\smallskip
Suppose that $f_{i_1}(x^{11})=0$. By (2) we have $f(x^{01})=x^{01}+e_{i_1}$, and since $x^{01}\in A$, we have $x^{01}+e_{i_1}\in A\cap X^{11}$. By~(1), $\Gamma$ has a geodesic path from $x^{01}+e_{i_1}$ to $x^{11}$. Hence it has a shortest geodesic path $P$ from $x^{01}+e_{i_1}$ to a state $x\in X^{11}$ such that $f_{i_1}(x)=0$. Let $P'$ be obtained from $P$ by adding the arc from $x^{01}$ to $x^{01}+e_{i_1}$ and from $x$ to $x+e_{i_1}$. By \cref{lem:R10b}, $G\setminus i_2$ has a negative walk from $i_1$ to itself. Hence $G\setminus i_2$ has a negative cycle and since $i_1$ is a feedback vertex set of $G\setminus i_2$, this negative cycle contains $i_1$. We prove similarly the second assertion.
\end{quote}

\begin{quote}
(6) {\em If $f_{i_1}(x^{11})=0$, then $G$ has a negative path from $i_2$ to $i_1$, and if $f_{i_2}(x^{11})=0$, then $G$ has a negative path from $i_1$ to $i_2$.}

\smallskip
Suppose that $f_{i_1}(x^{11})=0$. By (2) we have $f(x^{10})=x^{10}+e_{i_2}$, and since $x^{10}\in A$, we have $x^{10}+e_{i_2}\in A\cap X^{11}$. By~(1), $\Gamma$ has a geodesic path from $x^{10}+e_{i_2}$ to $x^{11}$. Hence it has a shortest geodesic path $P$ from $x^{10}+e_{i_2}$ to a state $x\in X^{11}$ such that $f_{i_1}(x)=0$. Let $P'$ be obtained from $P$ by adding the arc from $x^{10}$ to $x^{10}+e_{i_2}$ and from $x$ to $x+e_{i_1}$. By \cref{lem:R10b}, $G$ has a negative walk $W$ from $i_2$ to $i_1$ such that, denoting $j_1,\dots,j_\ell$ the vertex sequence of $W$ (thus $j_1=i_2$ and $j_\ell=i_1$), we have $i_1,i_2\not\in\{j_2,\dots,j_{\ell-1}\}$. Since $G\setminus\{i_1,i_2\}$ is acyclic, the vertex sequence has no repetition. Hence $W$ corresponds to a path. We prove similarly the second assertion.
\end{quote}

Since $x^{11}$ is not a fixed point, by (1) we have $f_{i_1}(x^{11})=0$ or $f_{i_2}(x^{11})=0$. If $f_{i_1}(x^{11})=0$ and $f_{i_2}(x^{11})=0$ then, by (3)-(6), $K^\pm_2$ is embedded in $G$ and so is $H_2$. So suppose, without loss, that $f_{i_2}(x^{11})=1$, and thus $f_{i_1}(x^{11})=0$ (since otherwise, by (1), $x^{11}$ is a fixed point, a contradiction). By (3)-(6), it only remains to prove that $G$ has a negative path from $i_1$ to $i_2$.

\medskip
Suppose, for a contradiction, that all the paths from $i_1$ to $i_2$ are positive. Let $I$ be the vertices that belong to a path from $i_1$ to $i_2$; by (4) there is at least one path from $i_1$ to $i_2$ thus $I$ is not empty and $i_1,i_2\in I$. Let $H$ be obtained from $G[I]$ by removing all the incoming arcs of $i_1$ and all the out-going arcs of $i_2$.

\begin{quote}
(7) {\em There is $L\subseteq I\setminus \{i_1,i_2\}$ such that the $L$-switch of $H$ is full-positive.}

\smallskip
Let $H'$ be obtained from $H$ by adding a positive arc from $i_2$ to $i_1$. Let $j_1,j_2$ be two vertices in $H'$. Then $H$ has a path from $i_1$ to $j_2$ and a path from $j_1$ to $i_2$. Since $H'$ has an arc from $i_2$ to $i_1$, it has a path from $j_1$ to $j_2$. So $H'$ is strongly connected. Suppose that $H'$ has a negative cycle $C$. Since $G\setminus \{i_1,i_2\}$ is acyclic, $H$ is acyclic, and thus $C$ contains the positive arc from $i_2$ to $i_1$. Hence the path of $C$ from $i_1$ to $i_2$ is negative, and since it is in $G$ we obtain a contradiction. Hence $H'$ has only positive cycles. Since $H'$ is strong, by \cref{pro:harary}, there is $L\subseteq I$ such that the $L$-switch of $H'$ is full-positive, and the $(I\setminus L)$-switch of $H'$ is also full-positive. If $i_1,i_2\not\in L$ then we are done. Otherwise, since the arc from $i_2$ to $i_1$ is positive in $H'$, we have $i_1,i_2\in L$, and we are done with $(I\setminus L)$ instead of $L$.
\end{quote}

Hence, up two a $L$-switch of $G$ and $f$ with $i_1,i_2\not\in L$, we can suppose that $H$ is full-positive, and since $i_1,i_2\not\in L$, $B$ still intersects $X^{00}$ and $A$ still intersects $X^{01},X^{10},X^{11}$ (but $x^{00}$ is no longer necessarily equal to $\ZERO$). Let $R$ be the set of vertices reachable from $i_1$ in $G$; so $I\subseteq R$. Let $J=R\setminus I$ and $K=V\setminus R$. Note that $G$ has no arc from $J$ to $I\setminus i_1$ (if there is an arc from $j \in J$ to $I\setminus i_1$, then $j$ is in a path from $i_1$ to $i_2$ so it belongs to $I$, a contradiction). Let   $\preceq$ be the partial order on $\B^V$ defined by $x\preceq y$ if and only if $x_{I\setminus i_1}\leq y_{I\setminus i_1}$ and $x_K=y_K$.

\begin{quote}
(8) {\em If $x\preceq y$ and $x_{i_1}\leq y_{i_1}$ then $f(x)\preceq f(y)$.}

\smallskip
Suppose that $x\preceq y$ and $x_{i_1}\leq y_{i_1}$. Since $\Delta(x,y)\subseteq R$ and $G$ has no arc from $R$ to $K$ we have $f(x)_K=f(y)_K$. Let $z$ be the configuration on $V$ defined by $z_I=y_I$ and $z_{V\setminus I}=x_{V\setminus I}$. We have $\Delta(x,z)\subseteq I$ and $x\leq z$. Given any $i\in I\setminus i_1$, since every arc of $G$ from a vertex in $I$ to $i$ is positive, we deduce that $f_i(x)\leq f_i(z)$. Since $\Delta(z,y)\subseteq J$ and $G$ has no arc from $J$ to $i$, we have $f_i(z)=f_i(y)$ and thus $f_i(x)\leq f_i(y)$.
\end{quote}

Since $\Gamma[A]$ has a path from $X^{01}$ to $X^{10}$, it has a path $P$ from $X^{01}$ to $X^{10}$ whose internal configurations are all in $X^{11}$. Let $y^0,\dots,y^{\ell+1}$ be the configurations of $P$ in order, and let $j_0,\dots,j_\ell$ be the direction sequence of $P$. Since $y^k\in X^{11}$ for $1\leq k\leq\ell$, we have $y^0\in X^{01}$ and $y^1\in X^{11}$ thus $j_0=i_1$. Similarly, since $y^\ell\in X^{11}$ and $y^{\ell+1}\in X^{10}$, we have $j_\ell=i_2$.

\medskip
Let $x^0=y^0$ and, for $1\leq k\leq\ell+1$, let $x^k=x^{k-1}+e_{j_k}$ if $\Gamma$ has an arc from $x^{k-1}$ in the direction $j_k$, and $x^k=x^{k-1}$ otherwise. Hence $\Gamma$ has a path from $x^0$ to $x^\ell$ whose direction sequence is a subsequence of $j_1,\dots,j_\ell$.

\medskip
We have $x^0_{i_1}=y^0_{i_1}=0$ and $y^1_{i_1}=1$. Since $i_1\not\in\{j_1,\dots,j_\ell\}$, we deduce that $x^k_{i_1}<y^{k+1}_{i_1}$ for all $0\leq k\leq \ell$. Hence we have $x^k\preceq y^{k+1}$ for all $0\leq k\leq \ell$. Indeed, for $k=0$ we have $x^0=y^0\preceq y^0+e_{i_1}=y^1$, and if $x^{k-1}\preceq y^k$ for some $1\leq k<\ell$, then, since $x^{k-1}_{i_1} <y^k_{i_1}$, we have $f(x^{k-1})\preceq f(y^k)$ by (8) and we deduce that $x^k\preceq y^{k+1}$. In particular $x^\ell\preceq y^{\ell+1}$. Since $y^{\ell+1}\in X^{10}$, we have $x^\ell_{i_2}=0$ and thus $x^\ell\in X^{00}$ because $x^\ell_{i_1}=0$. Since $\Gamma$ has a path from $x^0$ to $x^\ell$ and $x^0\in A$, we have $x^\ell\in A\cap X^{00}$, a contradiction. This proves that $G$ has a negative path from $i_1$ to $i_2$.
\end{proof}

If $G$ has feedback number at least $3$, then $H_2$ is not necessarily embedded in $G$ as illustrated by the following example.

\begin{example}\label{ex:fn3}
Let $f\in F(3)$ be defined by $f_1(x)=\bar x_3x_1\lor \bar x_3x_2$, $f_2(x)=\bar x_1x_2\lor \bar x_1x_3$ and $f_3(x)=\bar x_2x_3\lor \bar x_2x_1$. Then $\Gamma(f)$ is non-separating, and $H_2$ is not embedded in $G(f)$ since all the paths from $1$ to $3$ are positive, all the paths from $3$ to $2$ are positive and all the paths from $2$ to $1$ are positive.
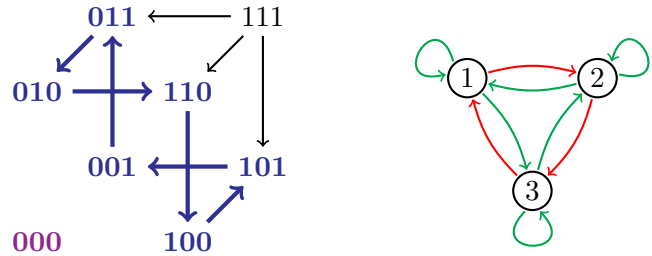
\[
\begin{array}{c}
\begin{tikzpicture}
\pgfmathparse{1}
\node (000) at (0,0){{\boldmath \textcolor{Plum}{$000$} \unboldmath}};
\node (001) at (1,1){{\boldmath \textcolor{Blue}{$001$} \unboldmath}};
\node (010) at (0,2){{\boldmath \textcolor{Blue}{$010$} \unboldmath}};
\node (011) at (1,3){{\boldmath \textcolor{Blue}{$011$} \unboldmath}};
\node (100) at (2,0){{\boldmath \textcolor{Blue}{$100$} \unboldmath}};
\node (101) at (3,1){{\boldmath \textcolor{Blue}{$101$} \unboldmath}};
\node (110) at (2,2){{\boldmath \textcolor{Blue}{$110$} \unboldmath}};
\node (111) at (3,3){$111$};
\path[thick,->,draw,black]
(010) edge[ultra thick,Blue] (110)
(110) edge[ultra thick,Blue] (100)
(100) edge[ultra thick,Blue] (101)
(101) edge[ultra thick,Blue] (001)
(001) edge[ultra thick,Blue] (011)
(011) edge[ultra thick,Blue] (010)
(111) edge (011)
(1110) edge (101)
(1110) edge (110)
;
\end{tikzpicture}
\end{array}
\qquad
\begin{array}{c}
\begin{tikzpicture}
\node[outer sep=1,inner sep=2,circle,draw,thick] (1) at ({-120-90}:1){$1$};
\node[outer sep=1,inner sep=2,circle,draw,thick] (2) at ({120-90}:1){$2$};
\node[outer sep=1,inner sep=2,circle,draw,thick] (3) at (-90:1){$3$};
\draw[Green,->,thick] (1.{-120-90-20}) .. controls ({-120-90-20}:2) and ({-120-90+20}:2) .. (1.{-120-90+20});
\draw[Green,->,thick] (2.{120-90-20}) .. controls ({120-90-20}:2) and ({120-90+20}:2) .. (2.{120-90+20});
\draw[Green,->,thick] (3.{-90-20}) .. controls ({-90-20}:2) and ({-90+20}:2) .. (3.{-90+20});
\path[->,thick]
(1) edge[red,bend left=15] (2)
(2) edge[red,bend left=15] (3)
(3) edge[red,bend left=15] (1)
(1) edge[Green,bend left=15] (3)
(3) edge[Green,bend left=15] (2)
(2) edge[Green,bend left=15] (1)
;
\end{tikzpicture}
\end{array}
\]
\end{example}

Note that the signed digraph of the example has positive feedback number equal to three.
\cref{ex:not-sep} shows that when $G$ is non-separating and has positive feedback number equal to two then $H_2$ is not necessarily embedded in $G$.
The following example shows that this is not necessarily the case even when adding the requirement that $G$ is strongly connected.

\begin{example}
\label{ex:strong-h2-not-embedded}
Consider $f\in F(4)$ defined by $f_1(x)=x_3 \vee x_1 \bar x_2$, $f_2(x)= x_4 \vee x_2 \bar x_1$, $f_3(x)= x_2 \bar x_3$ and $f_4(x)= x_1$. Then $\Gamma(f)$ is non-separating as shown in the figure below. $G(f)$ is strongly connected, has feedback number three and positive feedback number two. $H_2$ is not embedded in $G(f)$: vertices 1 and 2 are the only vertices that belong to disjoint positive cycles, and all negative cycles that contain one of them contain both.
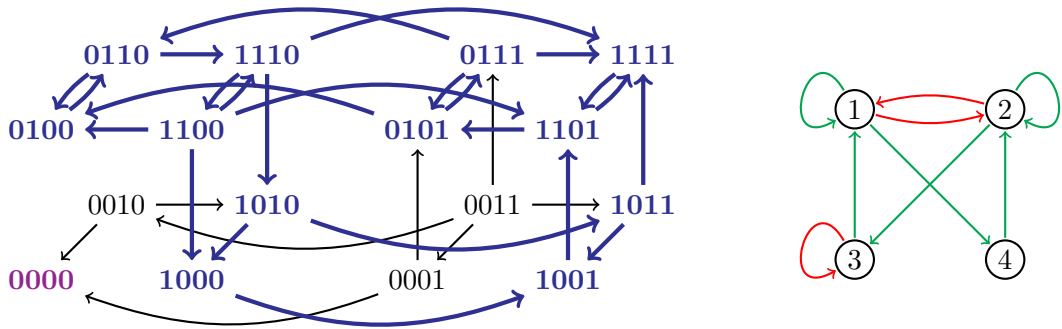
\[
\begin{array}{c}
\begin{tikzpicture}
\pgfmathparse{1}
\node (0000) at (0,0){{\boldmath \textcolor{Plum}{$0000$} \unboldmath}};
\node (0010) at (1,1){$0010$};
\node (0100) at (0,2){{\boldmath \textcolor{Blue}{$0100$} \unboldmath}};
\node (0110) at (1,3){{\boldmath \textcolor{Blue}{$0110$} \unboldmath}};
\node (1000) at (2,0){{\boldmath \textcolor{Blue}{$1000$} \unboldmath}};
\node (1010) at (3,1){{\boldmath \textcolor{Blue}{$1010$} \unboldmath}};
\node (1100) at (2,2){{\boldmath \textcolor{Blue}{$1100$} \unboldmath}};
\node (1110) at (3,3){{\boldmath \textcolor{Blue}{$1110$} \unboldmath}};
\node (0001) at (5,0){$0001$};
\node (0011) at (6,1){$0011$};
\node (0101) at (5,2){{\boldmath \textcolor{Blue}{$0101$} \unboldmath}};
\node (0111) at (6,3){{\boldmath \textcolor{Blue}{$0111$} \unboldmath}};
\node (1001) at (7,0){{\boldmath \textcolor{Blue}{$1001$} \unboldmath}};
\node (1011) at (8,1){{\boldmath \textcolor{Blue}{$1011$} \unboldmath}};
\node (1101) at (7,2){{\boldmath \textcolor{Blue}{$1101$} \unboldmath}};
\node (1111) at (8,3){{\boldmath \textcolor{Blue}{$1111$} \unboldmath}};
\path[thick,->,draw,black]
(0001) edge (0101)
(0001) edge[bend left=20] (0000)
(0010) edge (0000)
(0010) edge (1010)
(0011) edge (0001)
(0011) edge[bend left=20] (0010)
(0011) edge (0111)
(0011) edge (1011)
(0100) edge[Blue,ultra thick,bend right=10] (0110)
(0101) edge[Blue,ultra thick,bend right=20] (0100)
(0101) edge[Blue,ultra thick,bend right=10] (0111)
(0110) edge[Blue,ultra thick,bend right=10] (0100)
(0110) edge[Blue,ultra thick] (1110)
(0111) edge[Blue,ultra thick,bend right=10] (0101)
(0111) edge[Blue,ultra thick,bend right=20] (0110)
(0111) edge[Blue,ultra thick] (1111)
(1000) edge[Blue,ultra thick,bend right=20] (1001)
(1001) edge[Blue,ultra thick] (1101)
(1010) edge[Blue,ultra thick] (1000)
(1010) edge[Blue,ultra thick,bend right=20] (1011)
(1011) edge[Blue,ultra thick] (1001)
(1011) edge[Blue,ultra thick] (1111)
(1100) edge[Blue,ultra thick] (0100)
(1100) edge[Blue,ultra thick] (1000)
(1100) edge[Blue,ultra thick,bend right=10] (1110)
(1100) edge[Blue,ultra thick,bend left=20] (1101)
(1101) edge[Blue,ultra thick] (0101)
(1101) edge[Blue,ultra thick,bend right=10] (1111)
(1110) edge[Blue,ultra thick,bend right=10] (1100)
(1110) edge[Blue,ultra thick,bend left=20] (1111)
(1110) edge[Blue,ultra thick] (1010)
(1111) edge[Blue,ultra thick,bend right=10] (1101)
;
\end{tikzpicture}
\end{array}
\qquad
\begin{array}{c}
\begin{tikzpicture}
\node[outer sep=1,inner sep=2,circle,draw,thick] (1) at (0,0){$1$};
\node[outer sep=1,inner sep=2,circle,draw,thick] (2) at (2,0){$2$};
\node[outer sep=1,inner sep=2,circle,draw,thick] (3) at (0,-2){$3$};
\node[outer sep=1,inner sep=2,circle,draw,thick] (4) at (2,-2){$4$};
\draw[red,->,thick] (3) to [out=120,in=210,looseness=6] (3);
\draw[Green,->,thick] (1) to [out=120,in=210,looseness=6] (1);
\draw[Green,->,thick] (2) to [out=60,in=-30,looseness=6] (2);
\path[->,thick]
(1) edge[red,bend right=15] (2)
(1) edge[Green] (4)
(2) edge[red,bend right=15] (1)
(2) edge[Green] (3)
(3) edge[Green] (1)
(4) edge[Green] (2)
;
\end{tikzpicture}
\end{array}
\]
\end{example}

\section{Open problems}\label{sec:open-problems}

In this work we introduced notions of ``separation'' for asynchronous attractors of Boolean networks.
The mildest notion of separation requires that the minimal subspaces containing different attractors do not intersect.
If the minimal subspaces are also closed with respect to the dynamics, we speak of trap-separation.

\medskip
Previous results examined the cases of interaction graphs with no cycles, no positive cycles, or no negative cycles. In all cases, the attractors have strong separation properties, in particular they can always be separated using trap spaces.
Here we showed that the existence of at most one positive cycle (\cref{thm:PFN}) or at most one negative cycle (\cref{thm:NEGATIVE}) is sufficient for separation but not for trap-separation (\cref{ex:not-trap-sep}).
Separation is also guaranteed if there is no intersection between positive or negative cycles (\cref{thm:sep}), and trap-separation if there are no paths from negative to positive cycles (\cref{thm:trap-sep}).

\medskip
If we add the requirement that the graph is strongly connected, then trap-separation holds if there is at most one negative cycle (\cref{thm:NEGATIVE}). The theorem actually shows that a stronger property holds: the minimal subspaces containing attractors are trap spaces. Informally, we can say that the attractors are well approximated by minimal trap spaces. Whether the existence of at most one positive cycle in a strongly connected graph also implies conditions stronger than separation remains an open question. \cref{ex:strong-not-trapping} shows that under these conditions the attractors are not necessarily well approximated by minimal trap spaces.

\medskip
We also looked at how the existence of several attractors affects the existence of separate cycles of opposite signs.
We saw that if a network has multiple attractors, and at least one of them is not a fixed point, if it is trap-separating then its interaction graph must have disjoint positive and negative cycles (\cref{pro:disjoint-opposite}), but this is not necessarily true if the dynamics is only separating (\cref{ex:sep-not-conv-fix}).
For non-separating dynamics, all the examples we provided have disjoint positive and negative cycles, and at least one positive cycle with all vertices belonging to a negative cycle. We formulate the following conjecture.

\begin{conjecture}
Suppose that $G$ is non-separating. Then
\begin{itemize}
\item $G$ has a negative and a positive cycle that are disjoint, and 
\item $G$ has a positive cycle $C$ such that every vertex of $C$ belongs to a negative cycle.
\end{itemize}
\end{conjecture}

We saw that non-separating graphs have feedback number at least two (\cref{thm:PFN}), and that if the feedback number is exactly two then the graph contains $H_2$ (in the sense of \cref{thm:fvs2}). Starting from $H_2$ we can construct a non-separating strongly connected graph with $n\geq 3$ vertices by replacing the positive arc from $1$ to $2$ with a full-positive path with $n-2$ internal vertices (for instance, the signed digraph in \cref{ex:negative_feedback_1} is obtained by replacing the positive arc from $1$ to $2$ with a full-positive path with one internal vertex).

\begin{example}\label{ex:H2-n}
For each $n\geq 3$, let $f\in F(n)$ be defined by 
  \[f_1(x) = x_1 + x_2, \ \ f_2(x) = \bar{x}_1x_2 \vee x_n, \ \ f_3(x) = x_1, \ \ f_k(x) = x_{k-1} \text{ for } k=4,\dots,n.\]
The interaction graph of $f$ is as follows (the dotted green arrow represents a full-positive path).  
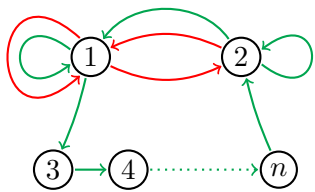
\[
\begin{tikzpicture}
\useasboundingbox (-2.2,-1.8) rectangle (2.2,0.7);
\node[outer sep=1,inner sep=2,circle,draw,thick] (1) at ({180}:1){$1$};
\node[outer sep=1,inner sep=2,circle,draw,thick] (2) at ({0}:1){$2$};
\node[outer sep=1,inner sep=2,circle,draw,thick] (3) at (-1.5,-1.5){$3$};
\node[outer sep=1,inner sep=2,circle,draw,thick] (4) at (-0.5,-1.5){$4$};
\node[outer sep=1,inner sep=2,circle,draw,thick] (n) at (1.5,-1.5){$n$};
\draw[Green,->,thick] (1.{180-20}) .. controls ({180-20}:2.3) and ({180+20}:2.3) .. (1.{180+20});
\draw[red,->,thick] (1.{180-60}) .. controls ({180-30}:2.8) and ({180+30}:2.8) .. (1.{180+60});
\draw[Green,->,thick] (2.{0-20}) .. controls ({0-20}:2.3) and ({0+20}:2.3) .. (2.{0+20});
\path[->,thick]
(1) edge[red,bend right=25] (2)
(1) edge[Green,bend left=5] (3)
(n) edge[Green,bend left=5] (2)
(2) edge[red,bend right=25] (1)
(2) edge[Green,bend right=55] (1)
(3) edge[Green] (4)
(4) edge[Green,dotted] (n)
;
\end{tikzpicture}
\]
It is strongly connected, has $n+5$ arcs and 7 cycles, of which 4 are positive and 3 are negative. This signed digraph is non-separating since $\Gamma(f)$ is non-separating. Indeed, $\ZERO$ is the unique fixed point of $f$. Furthermore, the set $T$ of configurations $x$ with $x_1=1$ or $x_2=1$ is a trap set which thus contains an attractor $A$. We can easily check that $\Gamma(f)$ has a path from any configuration in $T$ to $\ONE$, and thus $\ONE\in A$. Since $\Gamma(f)$ has a path from $\ONE$ to $e_2$ (with direction sequence $1,3,4\dots,n$) and a path from from $e_2$ to $e_1$ (with direction sequence $1,2$), we have $e_1,e_2,\ONE\in A$, and thus $[A]=\B^n$, so $\Gamma(f)$ is not separating.
\end{example}

Non-separating signed digraphs that do not contain $H_2$ do exist (\cref{ex:fn3});
however we conjecture that signed digraphs derived from $H_2$ provide lower bounds for strongly connected non-separating signed digraphs in terms of number of arcs and number of cycles.

\begin{conjecture}
  Every non-separating strongly connected signed digraph with $n\geq 3$ vertices has at least $n+5$ arcs and at least 7 cycles. At least 4 cycles are positive and at least 3 are negative.
\end{conjecture}

For signed digraphs that are separating but not trap-separating we can suggest stricter bounds, based on the following example.

\begin{example}\label{ex:sep-not-trap-sep-n}
For each $n\geq 4$ let $f\in F(n)$ be defined by
  \[f_1(x) = \bar{x}_{n-1} \bar{x}_{n}, \ \ f_n(x)= x_n \lor x_1\bar x_2x_3,\ \ f_k(x) = x_{k-1} \text{ for } k=2,\dots,n-1.\]
The interaction graph of $f$ is as follows (the dotted green arrow represents a full-positive path). 
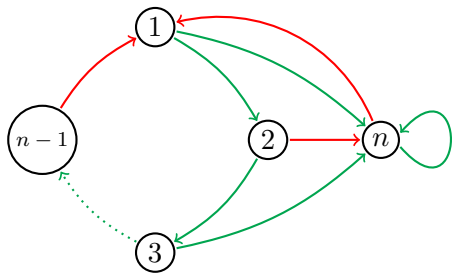
\[
\begin{tikzpicture}
\node[outer sep=1,inner sep=2,circle,draw,thick] (1) at (90:1.5){$1$};
\node[outer sep=1,inner sep=2,circle,draw,thick] (2) at ({0}:1.5){$2$};
\node[outer sep=1,inner sep=2,circle,draw,thick] (3) at (-90:1.5){$3$};
\node[outer sep=1,inner sep=2,circle,draw,thick] (n-1) at (180:1.5){\scriptsize  $n-1$};
\node[outer sep=1,inner sep=2,circle,draw,thick] (n) at (3,0){$n$};
\draw[Green,->,thick] (n.{0-20}) .. controls ({0-15}:4.3) and ({0+15}:4.3) .. (n.{0+20});
\path[->,thick]
(1) edge[Green,bend left=15] (2)
(2) edge[Green,bend left=15] (3)
(3) edge[Green,bend left=15,dotted]  (n-1)
(n-1) edge[red,bend left=15] (1)
(1) edge[Green,bend left=15] (n)
(2) edge[red] (n)
(3) edge[Green,bend right=15] (n)
(n) edge[red,bend right=40] (1)
;
\end{tikzpicture}
\]
It is separating since it has feedback number two ($\{1,n\}$ is a feedback vertex set) and no embedding of $H_2$. It is strongly connected, has $n+5$ arcs and 5 cycles, of which 2 are positive and 3 are negative. However, it is not trap-separating since $\Gamma(f)$ is not trap-separating. Indeed, $e_n$ is a fixed point of $f$ and $\Gamma(f)$ has an attractor $A$ whose configurations are $\sum_{i=1}^k e_i$ and $\ONE+\sum_{i=1}^k e_i$ for $k=1,\dots,n-1$. Thus $[A]=\{x_n=0\}$, and since $\Gamma(f)$ has an arc from $e_1+e_3$ to $e_1+e_3+e_n$ we have $\langle A\rangle=\B^n$, and thus $\Gamma(f)$ is not trap-separating.
\end{example}

\begin{conjecture}
  If a strongly connected signed digraph with $n\geq 4$ vertices is separating but not trap-separating, then it has at least $n+5$ arcs and at least 5 cycles, of which at least 2 are positive and at least 3 are negative.
\end{conjecture}

\section*{Acknowledgments}
The authors are grateful to Heike Siebert and the participants to the {\em Berlin Workshop on Theory and applications of Boolean interaction networks 2019} for helpful discussions.

Funding: Elisa Tonello was partially funded by the Volkswagen Stiftung (Volkswagen Foundation) under the initiative Life?—A fresh scientific approach to the basic principles of life (Project ID: 93063). Adrien Richard was supported by the Young Researcher project ANR-18-CE40-0002-01 ``FANs''.

\bibliographystyle{abbrv}
\bibliography{BIB}

\end{document}